\documentclass[english,aps,pra,amsfonts,amssymb,amsmath,twocolumn,groupedaddress,nofootinbib]{revtex4-1}
\usepackage[utf8]{inputenc}
\pdfpageattr{/Group <</S /Transparency /I true /CS /DeviceRGB>>}
\usepackage{amsthm}
\usepackage{amsmath}
\usepackage{amssymb}
\usepackage{subcaption}

\PassOptionsToPackage{normalem}{ulem}
\usepackage{ulem}

\usepackage{comment}
\usepackage{caption}
\usepackage{algorithm}
\usepackage{algorithmic}
\usepackage{setspace}

\usepackage[unicode=true,pdfusetitle,
 bookmarks=true,bookmarksnumbered=false,bookmarksopen=false,
 breaklinks=false,pdfborder={0 0 1},backref=false,colorlinks=false]
 {hyperref}
\usepackage[usenames,dvipsnames]{xcolor}
\hypersetup{colorlinks=true, linkcolor=Maroon, citecolor=OliveGreen, filecolor=magenta, urlcolor=Blue}

\makeatletter

\theoremstyle{plain}
\newtheorem{thm}{\protect\theoremname}
\theoremstyle{plain}
\newtheorem{prop}[thm]{Proposition}
\newtheorem{lem}[thm]{Lemma}

\newtheorem{observ}[thm]{\protect\observname}
\newtheorem{corol}[thm]{\protect\corolname}
\newtheorem{cnj}[thm]{Conjecture}
\newtheorem{example}[thm]{\protect\examplename}

\newtheorem{lemma}[thm]{Lemma}
\newtheorem{proposition}[thm]{Proposition}

\theoremstyle{definition}
\newtheorem{definition}[thm]{Definition}

\theoremstyle{remark}

\renewcommand{\phi}{\varphi}

\newcommand{\Z}{\mathbb Z}
\newcommand{\Q}{\mathbb Q}

\makeatother

\usepackage{babel}

\providecommand{\theoremname}{Theorem}
\providecommand{\definname}{Definition}
\providecommand{\observname}{Observation}
\providecommand{\corolname}{Corollary}
\providecommand{\examplename}{Example}
\providecommand{\problemname}{Problem}

\usepackage{pgfplots}
\usetikzlibrary{calc}
\usepackage{pgfplotstable}
\usepackage{colortbl}
\usepackage{booktabs}
\usepackage{pgfplotstable}

\newcommand{\tvect}[2]{%
  \ensuremath{\Bigl(\negthinspace\begin{smallmatrix}#1\\#2\end{smallmatrix}\Bigr)}}

\begin{document}

\title{Efficient Topological Compilation for Weakly-Integral Anyon Model}

\author{Alex Bocharov$^*$}
\author{Xingshan Cui$^\dagger$}
\author{Vadym Kliuchnikov$^*$}
\author{Zhenghan Wang$^{\bullet\dagger}$}

\affiliation{
$^*$Quantum Architectures and Computation Group, Microsoft Research, Redmond, WA (USA) \\
$^\dagger$  University of California, Santa Barbara, CA (USA) \\
$^\bullet$  Station Q, Microsoft Research, Santa Barbara, CA (USA)}

\begin{abstract}

A class of anyonic models for universal quantum computation based on weakly-integral anyons has been recently proposed.
While universal set of gates cannot be obtained in this context by anyon braiding alone, designing a certain type of sector charge measurement provides universality.

In this paper we develop a compilation algorithm to approximate arbitrary $n$-qutrit unitaries with asymptotically efficient circuits over the metaplectic anyon model. One flavor of our algorithm produces efficient circuits with upper complexity bound asymptotically in $O(3^{2\,n} \, \log{1/\varepsilon})$ and entanglement cost that is exponential in $n$. Another flavor of the algorithm produces efficient circuits with upper complexity bound in $O(n\,3^{2\,n} \, \log{1/\varepsilon})$ and no additional entanglement cost.
\end{abstract}

\maketitle

\section{Introduction}

Fault tolerance is becoming a key issue that will define success or failure of future programmable quantum computers.
Certain quasiparticles, called non-abelian anyons, provide a framework for coherent encoding of quantum information that will require little or no error correction.

Our primary goal is to propose an algorithm for efficient circuit synthesis (compilation) in one such non-abelian framework.

Braiding non-abelian objects such as anyons and zero-energy modes is the standard gate operation for topological quantum computation \cite{KitaevTop,FreedmanKitaev}.  But any physically realistic quantum operations are good for quantum information processing.  Besides braiding, measurement is a natural primitive for quantum computation.  While measurements in the quantum circuit model in the computational basis can always be postponed to the end, this cannot be done in topological quantum computation.  Therefore, we could gain extra computational power by supplementing braiding with measurements.  One physically realistic measurement in topological quantum computation is to measure the total charge of a group of anyons, which can be done by either projective measurement or interferometric measurement.

In \cite{CuiWang}, we pursue a qutrit generalization of the standard quantum circuit model.  Some anyon systems are very natural for the implementation of qutrits, e.g. anyons with quantum dimension $\sqrt{3}$.  One such anyon system is $SU(2)_4$---the first of the sequence of metaplectic anyons \cite{HNW}.  While braiding alone for $SU(2)_4$ is not universal as it is the case with the Majorana system, the metaplectic system is no longer like Majorana when measurement is added.  We proved that for $SU(2)_4$, braiding supplemented by projective measurement of the total charge of a pair of metapletic anyons is universal for qutrit quantum computation (see \cite{CuiWang}).

Our motivation for weakly-integral anyon framework is potential realization of metaplectic anyons and zero modes in physical systems.  Majoranas are closer to be well-controlled, but their computational power is impacted by the high complexity and cost of a universal basis \cite{SarmaFreNayak}.  Metaplectic models strike the right balance between controllability and universality.  There is some recent numerical evidence that $SU(2)_4$ might be realized in the $\nu=\frac{8}{3}$ fractional quantum Hall liquid (see \cite{PetersonEtAl}).
There is also recent research potentially leading to practical recipes for synthesizing and braiding parafermionic  zero modes in fractional quantum Hall liquids pared with $s$-wave superconductors (see \cite{ClarkeAlicea}). These are essentially recipes generalizing the synthesis of Majorana zero modes in the same general set up. In particular it is theoretically feasible that a species of $Z_4$-parafermion zero modes exhibiting $SU(2)_4$ statistics can be realized along these lines (ibid.).
Therefore, $SU(2)_4$ is a promising viable path to universal topological quantum computation.

In this paper we build upon the metaplectic model definition (\cite{CuiWang}) and develop algorithms for effective synthesis of efficient $n$-qutrit circuits over the model. Given a unitary target gate $U$ and an arbitrary small target precision $\varepsilon>0$ a circuit approximating $U$ to precision $\varepsilon$ is considered \emph{efficient} if the number of primitive gates in that circuit is asymptotically proportional to $\log{1/\varepsilon}$. An algorithm for synthesis of such efficient circuit is considered \emph{effective} if it can be completed on a classical computer in expected runtime that is polynomial in $\log{1/\varepsilon}$.

We develop two flavors of an effective general synthesis algorithm.
The first flavor makes a distinction between the parameter approximation cost and entanglement cost in an efficient circuit and produces such circuits with upper complexity bound in
$O(3^{2\,n}\,(\log_3{1/\varepsilon}+2 \, n+\log(\log(1/\varepsilon)))) +O((9\,(2+\sqrt{5}))^n)$. The second flavor makes no such distinction and produces efficient circuits with upper complexity bound in
$O(n\,3^{2\,n}\,(\log_3{1/\varepsilon}+2\, n+\log(\log(1/\varepsilon))))$. While the first flavor of our algorithm is clearly asymptotically superior when $n$ is fixed and $\varepsilon \rightarrow 0$, there is obviously a practical tradeoff threshold between the two flavors when $\varepsilon$ is fixed and $n$ is growing. Leading terms of our upper bounds for both complexities are expressed in terms of specific leading coefficients, not merely in the big $O$ terms.

The technique for the algorithm is number-theoretic in nature.
For any range of practically interesting precisions the circuits produced by our algorithms are significantly more efficient (both in the asymptotical and practical sense) than any hypothetical circuits obtainable by the Dawson-Neilsen version of Solovay-Kitaev algorithm (c.f. \cite{DN}).  Our algorithm designs are more broadly applicable to other classes of weakly-integral anyons involving the quantum dimension of $\sqrt{3}$.

The paper is organized as follows: in section II we make a very brief introduction into fundamental properties of metaplectic anyons, basic encodings and quantum gates; in section III the core circuit synthesis tools are developed, which are meant to reduce Householder reflections to axial reflections, and axial reflections are then described as metaplectic circuits in section IV. In section V and VI two approaches to synthesizing approximation circuits for arbitrary unitaries are introduced and compared, then the top level overview of the synthesis flow is given in section VII. Section VIII concludes the paper with some open problems and future work directions.

\section{Fusion, Braiding, and Basic Gates}

For completeness and readability we start with a very brief introduction into the concepts of braiding and fusion, focussing narrowly on the mathematical and logical side of these concepts. For a more broad exposure the reader is encouraged to look up the available tutorials on the subject such as  \cite{BeverlandPreskill}, \cite{KitaevTop},\cite{Preskill}.

\subsection{Background on fusion and braiding of non-Abelian anyons}

\emph{Anyons} are quasiparticles described by a certain Topological quantum field theory (TQFT), and, axiomatically such theory allows for a finite number of \emph{anyon species} that have distinct values $\{ \alpha, \, \beta, \, \gamma, \ldots \}$ of \emph{topological charge}. For example, one of the simplest theories leads to Fibonacci anyons that allows only two values of charge, $1$ and $\tau$, where $\tau$ is the charge of a non-trivial anyon and $1$ is the charge of \lq\lq no-anyon" or vacuum (\cite{Preskill}).

Given an ensemble of anyons $(a_1,\, a_2, \ldots, a_n)$ the structure of their collective state space $H$ depends on the underlying theory. If we measured collective topological charge of some subsequence of anyons in the ensemble, say $(a_i,\ldots, a_j), 1 \leq i < j \leq n$, the charge would probabilistically assume some value $c \in \{ \alpha, \, \beta, \, \gamma, \ldots \}$. After this is done, the state space of the ensemble is reduced to some smaller subspace
$H_{i,j,c} \subset H$. This is the phenomenon known as \emph{fusion} and the resulting topological charge is often called the \emph{fusion charge}.

Once we have measured out the fusion charge of several subsequences, we may end up with a one-dimensional state space, or, up to a global phase, with one specific state. This state can be characterized by the collection of measurement outcomes, and it is an established practice to represent such collection as a tree, called \emph{fusion tree}.

\begin{figure}[bt]
\includegraphics[width=3.5in]{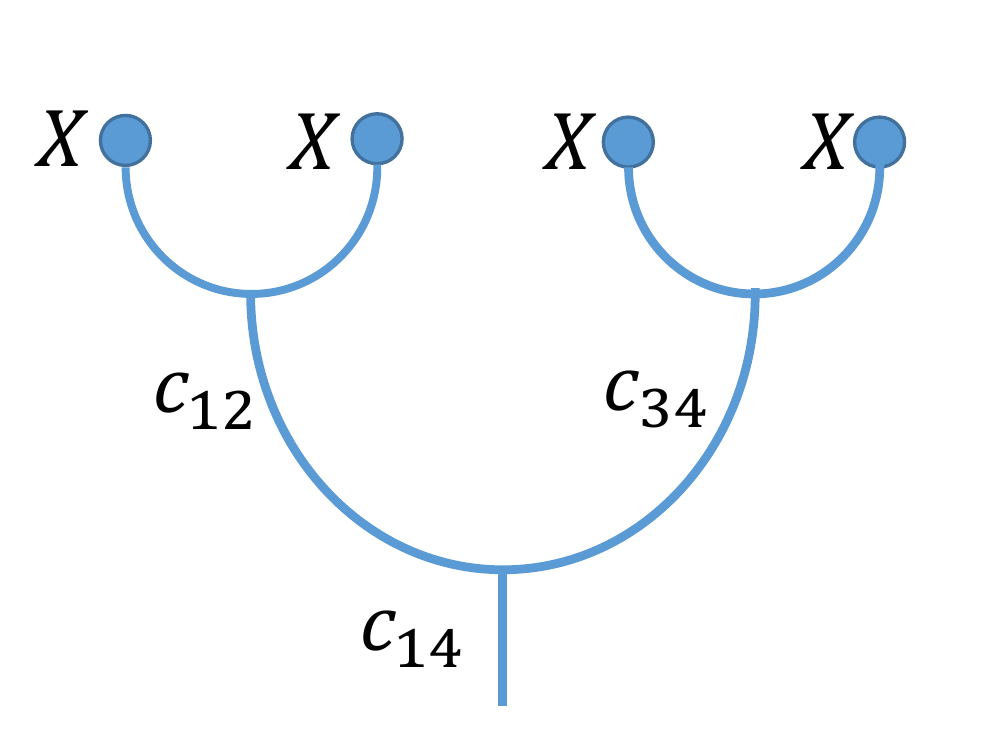}
\caption{\label{fig:fusion:tree:one} A fusion tree for anyonic quartet. The left pair of anyons has the fusion charge $c_{12}$, the right pair has the charge of $c_{34}$ and the overall charge of the quartet is $c_{14}$.}
\end{figure}

As a segway into the next subsection consider the following

\begin{example}
Theory of \emph{metaplectic anyons} allows five values of topological charge: $\{1,Z,X,X',Y\}$. Consider a quartet of anyons of type $X$, i.e. an ensemble $(a_1,a_2,a_3,a_4)$ where each anyon $a_i$ has the charge of $X$. Let us measure the charge $c_{12}$ of the pair $(a_1,a_2)$, then the charge $c_{34}$ the pair $(a_3,a_4)$ and then the charge $c_{14}$ of the entire quartet.
This sequence of measurement is represented by the tree shown in figure  \ref{fig:fusion:tree:one}.
\end{example}

Possible outcomes of fusion charge measurement are dictated by a set of \emph{fusion rules}.

A fusion rule has the following synthax:

\[ a \otimes b = \sum_{c}{N_{a\,b}^c \, c} \]

Here the left hand side stands for fusion of two systems with topological charges $a$ and $b$. The $\sum_{c}$ on the right is a disjunction indexed by all possible outcomes ($c$) of fusing of the two systems. $N_{a\,b}^c$ is the multiplicity of the corresponding outcome $c$. Its meaning being: if a pair of anyons of type $a$ and $b$ happened to have fused to the charge $c$ then their collective state space have reduced to an $N_{a\,b}^c$-dimensional Hilbert space.

\begin{example} \footnote{Incomplete set of rules is sufficient for our purposes.} \label{ex:mp:fusion:rules}
The following three rules are among the fusion rules of the  \emph{metaplectic anyon} theory:
 \begin{equation} \label{eq:mp:rule1}
 \forall c \in \{1,Z,X,X',Y\}, \, c \otimes 1 =  1 \otimes c = c
 \end{equation}
 \begin{equation} \label{eq:mp:rule2}
 X \otimes X = 1 \, + Y
 \end{equation}
 \begin{equation} \label{eq:mp:rule3}
 Y \otimes Y = 1 \, + Z \, + Y
 \end{equation}
\end{example}

To simplify the matters, we allow only multiplicities of $1$ below. Suppose $(a_1,\, a_2, \ldots, a_n)$ is an ensemble of anyons and a sequence of topological charge measurements has been selected that defines a certain fusion tree structure. Then the number of distinct fusion trees that are allowed by the fusion rules is precisely the dimension of the Hilbert state space $H$ of the ensemble, and there exists a basis in $H$ whose elements are labeled by those distinct fusion trees. We will describe a basis like this in the beginning of the next subsection.

While fusion bases are suitable for encoding the quantum information, natively topologically protected gates on such encodings can be derived from \emph{braiding} of non-Abelian anyons. Quite simply, \emph{braiding} is either an exchange of two distinct anyons in an ensemble or moving a single anyon along a complete closed loop. In general, braiding causes a non-trivial unitary action on the state space. By definition of "non-Abelian" these actions caused by different exchanges do not have to commute and the corresponding sets of unitary operators are not simultaneously diagonalizable. This creates the opportunity for building interesting and useful groups of unitary gates from braiding operations. Such groups are not always universal for quantum computation. Braiding happens to be universal in case of Fibonacci anyons (\cite{FreedmanKitaev}), and in the case of metaplectic anyons below universality can be achieved with a little help from measurement.

\subsection{Metaplectic basis and metaplectic circuits} \label{subsec:metaplectic:basis}


Metaplectic anyon model is defined in \cite{CuiWang} as an idealized multi-qutrit model, where each qutrit is encoded using a specific quartet of $SU(2)_4$ anyons and thus an $n$-qutrit quantum register is encoded using $4\,n$ anyons.
The model allows five values of topological charge $\{1,Z,X,X',Y\}$ and the relevant subset of fusion rules has been already listed in example \ref{ex:mp:fusion:rules}. We encode a standard qutrit using a quartet of anyons of type $X$ prepared such that their joint topological charge is $Y$. The corresponding basis states can be labeled by fusion trees such as shown on Fig. \ref{fig:fusion:tree:one} with the $c_{14}=Y$ constraint. It follows from the fusion rules that $(c_{12},c_{34}) \in \{ (1,Y),(Y,1),(Y,Y)\}$.

\begin{figure}[bt]
\includegraphics[width=3.5in]{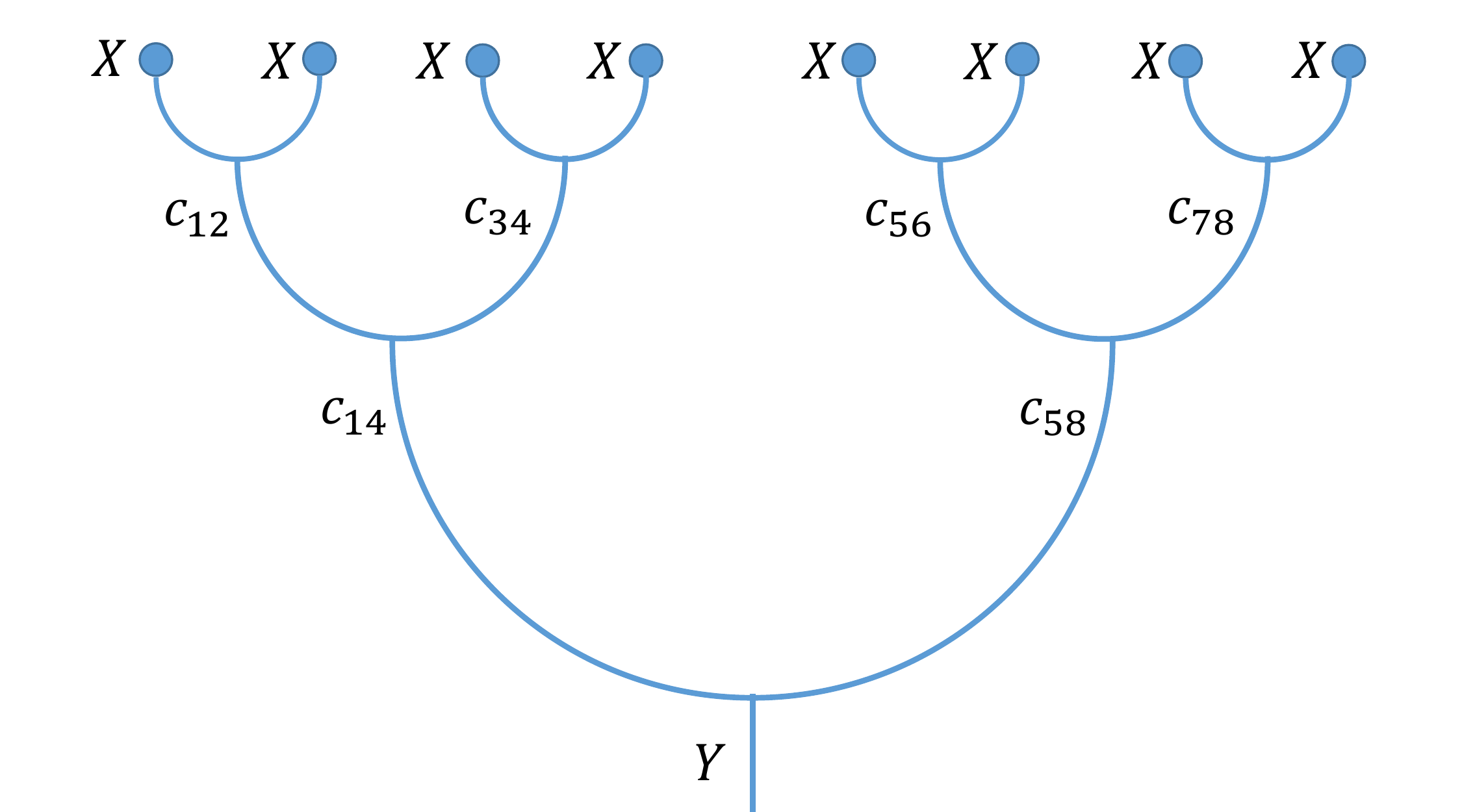}
\caption{\label{fig:fusion:tree:two} A fusion tree for $8$ anyons. The overall charge assumed to be $Y$. There are six fusion charges defining specific fused state.}
\end{figure}

One can do similar analysis on the state space $H$ of $8$ anyons of type $X$ prepared such that their overall topological charge is $Y$. The possible charges that label a basis in $H$ are shown on Fig. \ref{fig:fusion:tree:two}. Under the constraint $c_{14} = c_{58} = Y$ the system is reduced to a state in a $9$-dimensional subspace
$H' \subset H$ with an obvious ad hoc isomorphism of this subspace and $H_3 \otimes H_3$, where $H_3$ is the state space of the standard qutrit. We use $H'$ to encode a standard two-qutrit register and call it the
\emph{computational subspace}. It is not difficult to compute the dimension of $H$. As per fusion rules (\ref{eq:mp:rule1},\ref{eq:mp:rule2},\ref{eq:mp:rule3}) and by combinatorial enumeration, $\dim H = 21$. Thus $H'$ is a proper subspace of co-dimension $12$.

This analysis generalizes in a natural way to multi-qutrit encodings with more than two qutrits.

One should be cognizant that braiding of anyons from quartets encoding different qutrits (cf. Fig. \ref{fig:fusion:tree:two}) does not, in general, preserve the computational subspace, therefore we should only be deriving the multi-qutrit gates from the subgroup of braids that do preserve $H'$.

The actual derivation of primitive gates is beyond the scope of this paper. Below we summarize the designs developed in \cite{CuiWang}.

Consider the one-qutrit fusion basis $\{ |1,Y\rangle,|Y,1\rangle,|Y,Y\rangle\}$ introduced in the beginning of this subsection and relabel it as $\{|0\rangle = - |Y,Y\rangle, \, |1\rangle = |1,Y\rangle, \, |2\rangle = |Y,1\rangle\}$ (the minus sign leads to nicer algebra).

Introduce $\omega=e^{ 2 \pi \, i/3}$ and $\gamma=e^{\pi \, i/12}$.

Braiding of the anyons constituting a qutrit amounts to a finite-image representation of the braid group $B_4$ where the generators of $B_4$ are represented by the following unitaries in the above basis:

$\sigma_1 = \gamma\, diag(1,\omega,1)$, $\sigma_3 = \gamma\, diag(1,1,\omega)$,

$\sigma_2 =  \gamma^3 \, s_2; \, s_2 = \frac{1}{\sqrt{3}} \, \left(\begin{array}{ccc}
              1 & \omega & \omega \\
              \omega& 1 & \omega \\
              \omega & \omega & 1
            \end{array}\right)$

We observe that, up to global phase, $\sigma_1$ is equivalent to the $Q_1=diag(1,\omega,1)$ , $\sigma_3$ is equivalent to $Q_2=diag(1,1,\omega)$ and $\sigma_2$ is equivalent to the $s_2$.

For completeness we also need classical transpositions of the qutrit basis.

By direct computation, $\tau_{0,1} = i \, (\sigma_3 \, \sigma_2 \, \sigma_3)^2; \tau_{0,2} = i \, (\sigma_1 \, \sigma_2 \, \sigma_1)^2$ where $\tau_{j,k}$ is the $|j\rangle \leftrightarrow |k\rangle$ transposition.

Obviously $\tau_{0,1},\,\tau_{0,2}$ generate a faithful representation of the symmetric group $S_3$ on the qutrit and in particular, in terms of notations of \cite{CuiWang} we have $Q_0=\tau_{0,1} \, \sigma_1 \, \tau_{0,1}^{\dagger} = \tau_{0,2} \, \sigma_3 \, \tau_{0,2}^{\dagger}; \,\mbox{INC}= \tau_{0,2} \, \tau_{0,1}; \, \mbox{INC}^{\dagger}= \tau_{0,1} \, \tau_{0,2}$.

In the two-qutrit encoding explained above there is a certain braid explicitly composed out of $92$ anyon exchanges that preserves the computational subspace and, in the  $|j\rangle \otimes |k\rangle, \, j,k = 0,1,2$ basis implements the following  entangler:
\[\mbox{SUM} |j,k\rangle = |j,(j+k) \mod 3\rangle\]
\noindent which is a natural qutrit generalization of the $\mbox{CNOT}$.

It turns out that the gates designed above are not sufficient for the universal quantum computation, as per \cite{Gottesman}. They are known to generate a finite group that is projectively equivalent to the two-qutrit Clifford group.

However the \emph{reflection gate} \footnote{also called $\mbox{Flip}[2]$ gate elsewhere}
\[R_{|2\rangle} = diag(1,1,-1)\]

\noindent is outside the Clifford group and thus provides universality when added to the above gates.

The other two single-qutrit axial reflection operators are classically equivalent to $R_{|2\rangle}$:
$R_{|0\rangle} = \tau_{0,2} \, R_{|2\rangle} \, \tau_{0,2}^{\dagger}; \, R_{|1\rangle}=  \tau_{0,1} \,  R_{|0\rangle} \, \tau_{0,1}^{\dagger}$. We collectively call these reflections the $R$-\emph{gates}.

An  $R$-gate is implemented exactly via a certain \emph{measurement-assisted repeat-until-success circuit with two ancillary qutrits}, as described in \cite{CuiWang}, Lemma 5. The circuit performs a probabilistic protocol that succeeds in $3$ iterations on average (with the variance of the iterations to success equal to $6$). This is the most expensive protocol in our set so far \footnote{Although not nearly as expensive as a magic state distillation}, and for the purposes of resource estimation, we take the following

\emph{Assumption.} The cost of performing any braiding-only (generalized Clifford) gate, including the $\mbox{SUM}$ is trivial compared to the cost of performing an $R$-gate.

Therefore we will be using the $R$-\emph{count} as the measure of the cost of a quantum circuit.

\begin{definition}
A circuit composed of unitary gates introduced in this section is called a \emph{metaplectic circuit}.

The $R$-count of a metaplectic circuit is the minimal number of $R$-gates in all equivalent representations of the circuit.
\end{definition}

All the generators of metaplectic circuits are defined by matrices that are populated with algebraic numbers, and it follows from \cite{BG3} that the generator set is \emph{efficiently universal}, meaning that for any target unitary operator $G$ and small enough desired approximation precision $\varepsilon$ there exists a circuit of depth in $O(\log(1/\varepsilon))$ that approximates $G$ to precision $< \varepsilon$.

The main purpose of this paper is to develop actual classically feasible algorithm for finding such efficient approximating circuits.

\subsection{Useful additional gates.}
Here we expand the metaplectic basis defined in section \ref{subsec:metaplectic:basis} with additional useful gates.

1) $P$ \emph{gates}.

$P_j = I - (\omega^2+1) |j\rangle \langle j| = R_{|j\rangle} Q_j^2, j=0,1,2$

By design a $P$ gate has the $R$-count of 1. Any odd power of a $P$ gate also has the $R$-count of 1 while an even power of a $P$ gate has $R$-count of 0.

Here is a useful observation regarding the cost of $P$ gate sequences:
\begin{observ} \label{observ:two:P}
Any gate in the group generated by $\{P_0, P_1, P_2\}$ can be effectively represented as a product of the global phase in $\{\pm 1\}$ and a circuit of the $R$-count of at most $1$.
\end{observ}

\begin{proof}
Clearly $diag(-1,-1,-1)$ is identity up to the global phase of $(-1)$ and has the $R$-count of $0$.
Similarly each of the gates $f_{01} = diag(-1,-1,1), \, f_{02} = diag(-1,1,-1), f_{12} = diag(1,-1,-1)$ is an $R$ gate up to the global phase of $(-1)$ and has the $R$-count of $1$.

Now, any gate in the group generated by $\{P_0, P_1, P_2\}$ is of the form
$diag((-\omega^2)^{d_0},(-\omega^2)^{d_1},(-\omega^2)^{d_2})=
diag((-1)^{d_0},(-1)^{d_1},(-1)^{d_2}) \times diag(\omega^{2\,d_0},\omega^{2\,d_1},\omega^{2\,d_2})$. The second factor in this product has the $R$-count of $0$ by convention and the first factor is either $\pm I$ or one of the $R$ gates or one of the $f_{01}, \, f_{02},\, f_{12}$ gates and has the $R$-count of at most $1$.
\end{proof}

2) $\mbox{SWAP}$ gate.

While it is intuitively clear that the two-qutrit $\mbox{SWAP}$ gate can be performed by pure
braiding, a direct computation leads to the following

\begin{observ} \label{observ:SWAP}
$\mbox{SWAP} = $

$(\tau_{1,2} \otimes I) \mbox{SUM}_{1,2} \mbox{SUM}_{2,1} \mbox{SUM}_{2,1} \mbox{SUM}_{1,2}$
\end{observ}

Here $\tau_{1,2}$ is the single-qutrit transposition $|1\rangle \leftrightarrow |2\rangle$ (that can be expressed through already available transpositions as $\tau_{1,2}= \tau_{0,2} \tau_{0,1} \tau_{0,2}$).

By the usual notation convention here and everywhere the $\mbox{SUM}_{j,k}$ in multi-qutrit context is a shorthand for the two-qutrit sum gate applied to $j$-th qutrit as the control and $k$-th qutrit as the target (tensored with the identity gates on all other qutrits).

\bigskip

3) Axial reflection.

The following is key for our circuit synthesis:

\begin{definition}
Consider an integer $n\geq 1$ and let $|j\rangle, j = 0,\ldots,3^n-1$ be an element of standard $n$-qutrit basis.

The operator $R_{|j\rangle} = I^{\otimes n} -2 \, |j\rangle \langle j|$ is called an $n$-qutrit  \emph{axial reflection} (operator).
\end{definition}

Clearly it is indeed a reflection w.r.t. the hyperplane orthogonal to $|j\rangle$.

\section{Exact Single-Qutrit and Approximate Two-Level States.}

Consider the field of \emph{Eisenstein rationals} $\Q(\omega)$ which is a quadratic extension of $\Q$.
$\Z[\omega]$ is its integers ring called the ring of \emph{Eisenstein integers}. $\Z[\omega]$ has the group of units isomorphic to $\Z_6$ \ generated by $-\omega^2 = 1 + \omega$.

The two core tools needed for effective synthesis of metaplectic circuits are described in Lemmas \ref{lem:core:short:column} and \ref{lem:multi:two:level} below.

\begin{lemma} ["Short column lemma"] \label{lem:core:short:column}
Consider a unitary single-qutrit state $|\psi\rangle = (u \, |0\rangle + v \, |1\rangle + w \, |2\rangle)/\sqrt{-3}^L$ where
$u,v,w \in \Z[\omega]; L \in \Z$.

1) There is an effectively synthesizable metaplectic circuit $c$ with the $R$-count at most $L+1$ such that $c \, |\psi\rangle \in \{|0\rangle, |1\rangle, |2\rangle\}$.

2) The classical cost of finding such a circuit is linear in $L$.

\end{lemma}

Before proving the lemma, we need to handle one special case and make one algebraic observation.

\begin{lemma} [Special case.] \label{lem:zero:exponent}
If $|\psi\rangle$ is a unitary state the coefficients of which in computational basis are Eisenstein integers, then

1) One and only one coefficient is non-zero;

2) This non-zero coefficient is an Eisenstein integer unit;

3) $|\psi\rangle$ can be reduced to one of the computational basis states using at most one $P$ gate.
\end{lemma}

\begin{proof}
If $\psi_0, \ldots, \psi_N$ are the coefficients, then $\sum_{j=0}^{N} |\psi_j|^2 = 1$. Since for any $j$ , $|\psi_j|^2$ is a non-negative integer, all the coefficients, except one, some $\psi_{j_*}$, must be zeros, while $|\psi_{j_*}|^2=1$ and hence $\psi_{j_*}$ is a unit in $\Z[\omega]$. Therefore $\psi_{j_*} = (-\omega^2)^ d$ and
$(-\omega^2)^{-d \mod 6} \, \psi_{j_*} = 1$. Hence it is easy to find a $P$ gate of the form $G=I\otimes \ldots P_j^{-d \mod 6} \ldots \otimes I$ such that $G |\psi\rangle$ is a standard basis vector.
\end{proof}

Let us introduce the finite ring $\Z_3[\omega] = \Z[\omega]/(3\,\Z[\omega])$. This is a ring with exactly nine elements
$\{0,1,2,\omega, 2\, \omega, 1+\omega, 1+ 2\, \omega, 2+\omega, 2+ 2\, \omega\}$.

Let $\rho : \Z[\omega] \rightarrow \Z_3[\omega]$ be the natural epimorphism. By construction, its kernel consists of elements that are divisible by $3$.

Both the complex conjugation $*: \Z[\omega] \rightarrow \Z[\omega]$ and the norm map $|*|^2 : \Z[\omega] \rightarrow \Z$ can be consistently factored down to the
morphism $\tilde{*}: \Z_3[\omega] \rightarrow \Z_3[\omega]$ and the reduced norm map $\tilde{|*|^2} : \Z_3[\omega] \rightarrow \Z_3$
(since both $\rho \, *$ and $|*|^2 \mod 3$ annihilate the kernel of $\rho$).

For the benefit of several future constructions we need to analyze the action of the group of Eisenstein units $EU=\{-\omega^2\}$ on $\Z_3[\omega]$.

\begin{observ} \label{observ:orbits}
$\Z_3[\omega]$ is split into three orbits under the action of the group $EU$ as follows:

0) The one-element orbit $O_0$ of $0$; Note that $|0|^2=0$

1) The six-element orbit $O_1$ of $1$; Note that for any $z \in O_1$, $|z|^2 = 1 \mod 3$.

2) The two-element orbit $O_2$ of $1+2\,\omega$; Note that for any $z \in O_2$, $|z|^2 = 0 \mod 3$.
\end{observ}

This split is established by direct computations.

\begin{proof} (Of the "Short column lemma").
We will be proving the lemma by induction on $L$. For $L=0$ the claim follows from the lemma \ref{lem:zero:exponent}.

Consider a state with denominator exponent $L>0$.

Note that $\sqrt{-3}=1+2\, \omega$ and thus it is an Eisenstein integer. It follows, of course that $3 = -(1+2\,\omega)^2$ and thus $3$ is divisible by both
$1+2\,\omega$ and $(1+2\,\omega)^2$ in $\Z[\omega]$.

The state $|\psi\rangle$ is immediately reducible to a state of the form $1/\sqrt{-3}^{L-1} (u'\, |0\rangle+ v'\, |1\rangle + w'\, |2\rangle)$ if each of $u,v,w$ is divisible by $1+2\,\omega$
 and it is immediately reducible to a state of the form $1/\sqrt{-3}^{L-2}  (u''\, |0\rangle+ v''\, |1\rangle + w''\, |2\rangle)$ if each of $u,v,w$ is divisible by $3$ in $\Z[\omega]$.

From the unitariness condition on $|\psi\rangle$ we have $|u|^2+|v|^2+|w|^2=3^L$.
Given $L>0$, then $3^L \mod 3=0$ and thus $(|u|^2 \mod 3)+(|v|^2 \mod 3)+(|w|^2 \mod 3)=0$. By direct computation we check, however, that for any $z \in \Z[\omega]$, $|z|^2 \mod 3$ is either $0$ or $1$. By simple exclusion argument for $(|u|^2 \mod 3)+(|v|^2 \mod 3)+(|w|^2 \mod 3)=0$ to hold, either all the summands must be $0$ or all the summands must be $1$.

Let us distinguish the two cases.

Case $0$: $(|u|^2 \mod 3)=(|v|^2 \mod 3)=(|w|^2 \mod 3)=0$

As per the above observation \ref{observ:orbits} the residues $\rho(u), \rho(v), \rho(w)$ belong to the union of orbits $O_0$ and $O_2$.

In the edge case when all three belong to the orbit $O_0$ , each of the $u,v,w$ is divisible by $3$. As per earlier remark, $|\psi\rangle$ is reducible to the case of denominator
exponent $L-2$ and we do not need to apply any gates for this reduction.

More generally within the case $0$ each of the residues $\rho(u), \rho(v), \rho(w)$ is divisible by $\rho(1+2\,\omega)$. However if $\rho(z)$ is divisible by  $\rho(1+2\,\omega)$ then $z$ is divisible by $1+2\,\omega$ in the $\Z[\omega]$. Indeed , the divisibility of the residue implies that $z=(1+2\,\omega)\, z' + 3\,z'', \, z', z'' \in \Z[\omega]$,
but, as we noted, $3$ is divisible by $1+2\,\omega$ in the $\Z[\omega]$. Thus the general subcase allows reduction to the denominator
exponent $L-1$ without application of any gates.

Case $1$ : $(|u|^2 \mod 3)=(|v|^2 \mod 3)=(|w|^2 \mod 3)=1$.

We are going to find a short circuit $c_L$ of $R$-count at most $1$ such that $c_L \, |\psi\rangle$ is reduced to a case with denominator exponent at most $L-1$. (This would complete the induction step.)

Suppose first that $\rho(v)= \rho(w) = \omega^2 \, \rho(u) \in \Z_3[\omega]$,  which means that
$v=\omega^2 \,  u + 3 \, v', \, w = \omega^2 \, u  + 3 \, w'$ for some $v',w' \in \Z[\omega]$ and it follows that
$s_2 \, |\psi\rangle = (-(u + \omega \, v' + \omega \, w')\, |0\rangle - (v' + \omega \, w') \, |1\rangle - (\omega\, v'+w') \, |2\rangle))/\sqrt{-3}^{L-1}$.

Thus, in this particular special case the denominator exponent is reduced to $L-1$ by application of the single $s_2$ gate that has $R$-count $0$.

In general, since $(|\omega^2 \, u|^2 \mod 3)=(|u|^2 \mod 3)=(|v|^2 \mod 3)=(|w|^2 \mod 3)=1$, then  $\omega^2 \, \rho(u),\rho(v),\rho(w)$ must belong to the same orbit $O_1$ of the unit group $EU$. This means, in particular we can effectively find integers $d_v, d_w$ such that $\omega^2 \,  \rho(u)=\rho((-\omega^2)^{d_v}\,v)=\rho((-\omega^2)^{d_w}\,w)=r \in \Z_3[\omega]$.  Hence the short circuit $c_L = s_2\, P_1^{d_v} \, P_2^{d_w}$ reduces the state as shown .
As per the observation \ref{observ:two:P}, $P_1^{d_v} \, P_2^{d_w}$ in this circuit is equivalent to a circuit of $R$-count at most $1$ up to the possible global phase of $\pm 1$.

This completes the induction step.

\end{proof}

\begin{example}
Consider unitary column $|K\rangle = ((2 + i \,\sqrt{3})\,|0\rangle + |1\rangle+ |2\rangle )/3$.

$|K\rangle$ is reduced to basis state at $R$-count of $2$ as follows: $s_2 \,  R_{|0\rangle}\, Q_1^2\, Q_2^2 \, s_2 \, R_{|0\rangle}\, |K\rangle = |0\rangle$

Note that

$s_2 \,  R_{|0\rangle}\, Q_1^2\, Q_2^2 \, s_2 \, R_{|0\rangle} = -\omega \, \sigma_2 \,  R_{|0\rangle}\, \sigma_1^2 \, \sigma_3^2 \, \sigma_2 \,  R_{|0\rangle}$.

\end{example}

Below we present the method suggested by the lemma  \ref{lem:core:short:column} in algorithmic format

\begin{algorithm}[H]
\caption{Reduction of a short unitary column}
\label{alg:short:column:reduction}
\algsetup{indent=2em}
\begin{algorithmic}[1]
\REQUIRE{$L \in \Z$, $u,v,w \in \Z[\omega]$}
\STATE { $ret \gets \langle \mbox{empty} \rangle$ }
\WHILE {$L > 0$}
\STATE{$\{\nu u, \nu v, \nu w \}=\{ |u|^2, |v|^2, |w|^2\} \mod{3}$}
\IF {$\nu u =\nu v = \nu w = 1$}
\STATE{Find $d_v, d_w \in \{-2,-1,0,1,2,3\}$ such that}
\STATE{ $ \omega^2 \, u \equiv (-\omega^2)^{d_v} v \equiv  (-\omega^2)^{d_w} w \mod 3 $}
\STATE{$\{u,v,w\} \gets \{u,(-\omega^2)^{d_v} v,(-\omega^2)^{d_w} w\}$}
\STATE {$v' \gets (v - \omega^2\, u)/3; w' \gets (w - \omega^2\, u)/3$}
\STATE{$\{u,v,w\} \gets$ }
\STATE{$ \{- (u+ \omega \, v'+\omega \, w'), -(v' + \omega\, w'), -(\omega\, v' + w') \}$}
\STATE {$ret \gets s_2\, P_1^{d_v} \, P_2^{d_w} \, ret$}
\ELSE
\STATE {$\{u,v,w\} \gets \{u,v,w\}/(2\, \omega+1)$}
\ENDIF
\STATE{$L \gets L-1$}
\ENDWHILE
\STATE{ Implied $L=0$; Only one of $u,v,w$ is non-zero.}
\STATE{ Find classical $g$ s. t. $g (u |0\rangle + v |1\rangle + w |2\rangle) = u' |0\rangle$}
\STATE{ Find $d \in \{-2,-1,0,1,2,3\}$ such that $(-\omega^2)^d = u'$ }
\RETURN { $P_0^{-d} \, g \, ret$ }

\end{algorithmic}
\end{algorithm}

\begin{lemma} \label{lem:core:two:level:state}

Consider a "two-level" unitary single-qutrit state $|\phi\rangle =x \, |0\rangle + y \, |1\rangle + z \, |2\rangle$ where $x \, y \, z = 0$ and let $\varepsilon$ be an arbitrarily small positive number.

1) There is a family of effectively synthesizable states of the form
$|\psi_{\varepsilon}\rangle = (u_{\varepsilon} \, |0\rangle + v_{\varepsilon} \, |1\rangle + w_{\varepsilon} \, |2\rangle)/\sqrt{-3}^{L_{\varepsilon}}$;
$u_{\varepsilon},v_{\varepsilon},w_{\varepsilon} \in \Z[\omega]; L_{\varepsilon} \in \Z$ such that $|\psi_{\varepsilon}\rangle$ is an $\varepsilon$-approximation of $|\phi\rangle$ and
$L_{\varepsilon} \leq 4\,\log_3(1/\varepsilon)+O(\log(\log(1/\varepsilon)))$.

2) The expected average classical cost of finding each $|\psi_{\varepsilon}\rangle$ is polynomial in $\log(1/\varepsilon)$.

\end{lemma}

A proof of this lemma is found in Appendix \ref{sec:single:qutrit:approx}. The proof is very technical. It combines  elementary geometry with rather profound number theory, which is based on a mild number-theoretical hypothesis (conjecture \ref{conj:norm:eq:solvability}).

It follows from the two lemmas that a two-level unitary state can be prepared with precision $\varepsilon$ from a standard basis state using a metaplectic circuit of $R$-count at most $4\,\log_3(1/\varepsilon)+O(\log(\log(1/\varepsilon)))$ and in fact this readily generalizes to multiple qutrits as follows:

\begin{lemma} ["Two-level approximation lemma"] \label{lem:multi:two:level}
Consider an integer $n\geq 1$ and let $|\phi\rangle$ be a unitary $n$-qutrit state that has at most two non-zero components in the standard $n$-qutrit basis.

For arbitrarily small $\varepsilon>0$

1) There is an effectively synthesizable metaplectic circuit $c$ with the $R$-count at most $4\,\log_3(1/\varepsilon)+O(\log(\log(1/\varepsilon)))$ such that $c \, |0\rangle$ is an $\varepsilon$-approximation of $|\phi\rangle$.

2)  The expected average classical cost of finding  such a circuit is polynomial in $\log(1/\varepsilon)$.
\end{lemma}

Before proving the lemma we need two lesser technical facts that are useful in their own right:

\begin{lem} \label{lem:transitive:on:axial}
Let $|b_1\rangle$  and $|b_2\rangle$ be two standard $n$-qutrit basis states. There exists an effectively and exactly representable classical permutation $\pi$ such that $|b_2\rangle= \pi \, |b_1\rangle$
\end{lem}

\begin{proof}
In the case of $n=1$ the  $\Z_3$ group generated by  $\mbox{INC}$  acts transitively on the standard basis $\{|0\rangle,\, |1\rangle,\, |2\rangle\}$.

Consider $|b_k\rangle= |(b_k)_1,\ldots ,  (b_k)_n\rangle, \, n\geq 1, \, k=1,2$. Let $\pi_j \in \{I,\mbox{INC} , \mbox{INC}^2\}$ be such that $\pi_j |(b_1)_j\rangle = |(b_2)_j\rangle, j = 1,\ldots, n$.

Then $\pi = \otimes_{j=1}^{n} \pi_j$ is the desired permutation.
\end{proof}

\begin{lem} \label{lem:two:basis:vectors}
1) For any two standard $n$-qutrit basis vectors $|j\rangle$ and $|k\rangle$ there exists a classical effectively representable metaplectic gate $g$, such that for $|j'\rangle= g |j\rangle$ and $|k'\rangle= g |k\rangle$ we have $|j' - k'| < 3$.

2) Such a gate $g$ can be effectively represented with at most $(n-1)$ instances of the $\mbox{SUM}$,
 $\mbox{SUM}^{\dagger}$ or $\mbox{SWAP}$ gates.
\end{lem}

In other words, digital representations of $j'$ and $k'$ base $3$ are the same except possibly for the least-significant base-$3$ digit.

\begin{proof}
At $n=1$ there is nothing to prove.

Given Lemma \ref{lem:transitive:on:axial}, for $n=2$ the general pair of basis vectors can be reduced to the case where $|j\rangle = |00\rangle$. When $|k\rangle = |0,k_1\rangle$ no further transformations are needed, when $|k\rangle = |k_0,0\rangle$ a single $\mbox{SWAP}$ suffices. The remaining cases are covered by
$\mbox{SUM}_{2,1}^{\dagger} |11\rangle = \mbox{SUM}_{2,1} |21\rangle = |01\rangle$,
$\mbox{SUM}_{2,1} |12\rangle = \mbox{SUM}_{2,1}^{\dagger} |22\rangle = = |02\rangle$.

Suppose $n>2$ and the lemma has been proven for multi-qutrit vectors in fewer than $n$ qutrits.

Let $|j\rangle=|j_1\, \ldots, j_{n-1},  j_n\rangle , |k\rangle=|k_1\, \ldots, k_{n-1}, k_n\rangle$ be base-$3$ representations of the two vectors.

By induction hypothesis, one can effectively find $(n-1)$-qutrit classical metaplectic gate $g_{n-1}$ such that
$(g_{n-1} \otimes I) |j_1\, \ldots, j_{n-1},  j_n\rangle = |\ldots, j'_{n-1},  j'_n\rangle$ and
$(g_{n-1} \otimes I) |k_1\, \ldots, k_{n-1},  k_n\rangle= |\ldots, k'_{n-1},  k'_n\rangle$ may differ only at $(n-1)$-st and $n$-th position.

Select a two-qutrit classical gate $g_2$, as shown above,  such that $g_2 |j'_{n-1},  j'_n\rangle$ and $g_2 |k'_{n-1},  k'_n\rangle$ differ only in the last position.
Then, by setting $g=(I^{\otimes (n-2)}\otimes g_2)(g_{n-1} \otimes I)$ we complete the induction step.

\end{proof}

\begin{proof} (Of the two-level state approximation lemma.)
We start by reducing $|\phi\rangle$ to the form $x \, |a_1\ldots a_{n-1},d\rangle + z \, |a_1\ldots a_{n-1},f\rangle, \, a_1,\ldots,a_{n-1},d,f \in \{0,1,2\}$ using a classical circuit $b$ described in  Lemma \ref{lem:two:basis:vectors}.
Let $e \in \{0,1,2\}$ be the "missing" digit such that $\{d,e,f\}$ is a permutation of $\{0,1,2\}$.

Using Lemma \ref{lem:core:two:level:state}
 we can effectively approximate the single-qutrit state $x\,|d\rangle + z \, |f\rangle$ by an Eisenstein state of the form
$|\eta\rangle = (u\,|d\rangle + v\, |e\rangle+ w \, |f\rangle)/\sqrt{-3}^k, u,v,w \in \Z[\omega], k \in \Z$ to precision $\varepsilon$ with
$k \leq 4\, \log_3(1/\varepsilon)+O(\log(\log(1/\varepsilon)))$.

Using Lemma \ref{lem:core:short:column} we can effectively synthesize a single-qutrit metaplectic circuit $c_1$ with $R$-count at most $k+1$ such that $c_1 \, |0\rangle = |\eta\rangle$.

Let $c_n = (I^{\otimes(n-1)} \otimes c_1)$.

Clearly $b^{\dagger} \, c_n \,  |a_1\ldots a_{n-1},0\rangle$ is an $\varepsilon$-approximation of $|\phi\rangle$. But $|a_1\ldots a_{n-1},0\rangle$ can be prepared exactly from $|0\rangle$ using at most $n-1$ local $\mbox{INC}$ gates, which finalizes the desired circuit.
\end{proof}

\begin{corol} \label{corol:two:level:reflections}
Consider an integer $n\geq 1$ and let $|\phi\rangle$ be a unitary $n$-qutrit state that has at most two non-zero components in the standard $n$-qutrit basis and consider the corresponding Householder reflection operator $R_{|\phi\rangle} = I^{\otimes n} -2 \, |\phi\rangle \langle \phi|$.

For arbitrarily small $\varepsilon>0$

1) There is an effectively synthesizable metaplectic circuit $c$ with the $R$-count at most $4\,\log_3(1/\varepsilon)+O(\log(\log(1/\varepsilon)))$ such that  $ c \, R_{|\overline{0}\rangle} \, c^{\dagger}$ is a $\varepsilon$-approximation  of $R_{|\phi\rangle}$.(Where
$|\overline{0}\rangle = |0\rangle^{\otimes n}$.)

2)  The expected average classical cost of finding  such a circuit is polynomial in $\log(1/\varepsilon)$.

\end{corol}

\begin{proof}
As per \cite{Kliuchnikoff}, if the distance between state $|\phi\rangle$ and $|\psi\rangle$ is less than
$\varepsilon/(2\,\sqrt{2})$ , then the distance between $R_{|\phi\rangle}$ and $R_{|\psi\rangle}$ is less than
$\varepsilon$. Using Lemma \ref{lem:multi:two:level} one can effectively find a metaplectic circuit $c$ with the
$R$-count in $4\,\log_3(1/\varepsilon)+O(\log(\log(1/\varepsilon)))$ such that
$c \, |\overline{0}\rangle $ approximates $|\phi\rangle$ to precision $\varepsilon/(2\,\sqrt{2})$ and the corollary follows.
\end{proof}

This result applies in a straightforward manner to one-parameter special diagonal unitary:

\begin{corol} \label{corol:two:level:diagonal}
Consider an integer $n\geq 1$ and an $n$-qutrit diagonal operator of the form
$D= I^{\otimes n} + (e^{i\, \theta}-1)\, |j\rangle \langle j| + (e^{-i\, \theta} -1)\, |k\rangle \langle k|$ where
$j,k \in \{0,\ldots,3^n-1\}, j \neq k$.

For arbitrarily small $\varepsilon>0$ there is an effectively synthesizable circuit at distance  $< \varepsilon$ from $D$ composed out of at most two axial $n$-qutrit reflection operators and local metaplectic gates with the total $R$-count of

1) at most $8\,\log_3(1/\varepsilon)+O(\log(\log(1/\varepsilon)))$ when $n=1$, and

2) at most $16\,\log_3(1/\varepsilon)+O(\log(\log(1/\varepsilon)))$ when $n>1$.
\end{corol}

Indeed, the diagonal unitary of this form is equal to $r_1 \, r_2$ where
$r_1= I^{\otimes n} - |j\rangle \langle j| - |k\rangle \langle k| + |j\rangle \langle k| + |k\rangle \langle j|$,
$r_2= I^{\otimes n} - |j\rangle \langle j| - |k\rangle \langle k| + e^{ -i \, \theta}\, |j\rangle \langle k| + e^{ i \, \theta}\,|k\rangle \langle j|$ and both $r_1$ and $r_2$ are two-level reflection operators. We note that for $n=1$ the $r_1$ is a Clifford gate and has trivial cost.

Since mult-qutrit axial reflection are going to grow in importance below, we offer a decomposition method for them in the next section.

\section{Implementation of Axial Reflection Operators} \label{sec:axial:reflections}

Let $|b\rangle$ be a standard $n$-qutrit basis state.

Then an \emph{axial reflection operator} $R_{|b\rangle}$ is defined as

$R_{|b\rangle}=I^{\otimes n} - 2\,|b\rangle \langle b|$

Clearly, $R_{|b\rangle}$ is represented by a diagonal matrix that has a $-1$ on the diagonal in the position corresponding to $|b\rangle$ and $+1$ in all other positions.


As per Lemma \ref{lem:transitive:on:axial} any two axial reflection operators are equivalent by conjugation with an effectively and exactly representable classical permutation.
Since we consider the cost of classical permutations to be negligible compared to the cost of the $R$ gates, we hold that for a fixed $n$ all the $n$-qutrit axial reflection operators have essentially the same cost.

We are going to show in this section that all the $n$-qutrit axial reflection operators can be effectively and exactly represented.

In view of the above if suffices to represent just one such operator for each $n$.
We start with somewhat special case of $n=2$.

\begin{observ} \label{observ:CFlip}
The circuit

$(I \otimes R_{|0\rangle}) \, \mbox{SUM} (I \otimes R_{|1\rangle}) \, \mbox{SUM} (R_{|2\rangle} \otimes R_{|2\rangle}) \, \mbox{SUM}$

is an exact representation of $(-1) R_{|20\rangle}$
\end{observ}

This is established by direct matrix computation.

We are going to generalize this solution to arbitrary $n \geq 2$ and note that the occurrence of the global phase $(-1)$ is exceptional and happens only at $n=2$.

\begin{lemma}
Given $n>2$ , denote by $\bar{2}$ in the context of this lemma a string of $n-2$ occurrences of $2$.

Then the circuit

$c_{20\bar{2}}=$

$(I \otimes R_{|0\bar{2}\rangle}) \, \mbox{SUM}_{1,2} \, (I \otimes I \otimes R_{|\bar{2}\rangle}) \, (I \otimes R_{|1\bar{2}\rangle}) $ $ \mbox{SUM}_{1,2} \, \mbox{SWAP}_{1,2} \, (I \otimes R_{|2\bar{2}\rangle}) \, \mbox{SWAP}_{1,2} \,  (I \otimes R_{|2\bar{2}\rangle}) \, \mbox{SUM}_{1,2}$

is an exact representation of the operator $R_{|20\bar{2}\rangle}$.

\end{lemma}

\begin{proof}

Let $|b\rangle$ be an element of the standard $n$-qutrit basis.
The circuit consists of diagonal operators and three occurrences of $\mbox{SUM}_{1,2}$.
Let $|b_1b_2\bar{b}\rangle$ be the ternary representation of $|b\rangle$ where $\bar{b}$ stands for the substring of the $n-2$ least significant ternary digits of $b$.
It is almost immediate that the circuit $c_{20\bar{2}}$ represents a diagonal unitary.
Indeed, when the input is $|b_1b_2\bar{b}\rangle$ we can only get $\pm |b_1b_2\bar{b}\rangle$, $\pm |b_1\,\mbox{INC}\,b_2\bar{b}\rangle$ or $\pm |b_1\,\mbox{INC}^2b_2\bar{b}\rangle$, up to swap,  after applying each subsequent operator of the circuit, and clearly we can only get $\phi |b_1b_2\bar{b}\rangle, \, \phi=\pm 1$ after the entire circuit is applied.

The lemma claims that $\phi=-1$ if and only if $b=20\bar{2}$.

Consider the cases when $b_1=0$ or $b_1=1$. It is easy to see that, whatever is the value of  $b_2$, one and only one of the operators  $(I \otimes R_{|0\bar{2}\rangle}), (I \otimes R_{|1\bar{2}\rangle}), (I \otimes R_{|2\bar{2}\rangle})$ activates $R_{|\bar{2}\rangle}$ on $|\bar{b}\rangle$  and this activation always cancels out with $(I \otimes I \otimes R_{|\bar{2}\rangle})$ (since $R^2=$ identity for any reflection $R$). So the result is identity.

If $b_1=2, b_2 \neq 0$ the five rightmost operations of the circuit produce $|2\rangle \otimes (\mbox{INC}^2 |b_2\rangle) \otimes (R_{|\bar{2}\rangle} |\bar{b}\rangle)$, an action that is subsequently canceled out by $I \otimes I \otimes R_{|\bar{2}\rangle}$.
It is also easy to see that for $b_2=1$  or $b_2=2$ the remaining two reflections $R_{|0\bar{2}\rangle}$ and $R_{|1\bar{2}\rangle}$ amount to non-operations. Therefore the net result is identity.

We are left with the important case of $b_1=2, b_2=0$.

By definition, $\mbox{SUM}_{12} |20\bar{b}\rangle = |22\bar{b}\rangle$ and then the subsequence $\mbox{SWAP}_{1,2} \, (I \otimes R_{|2\bar{2}\rangle}) \, \mbox{SWAP}_{1,2} \,  (I \otimes R_{|2\bar{2}\rangle})$ activates operator $R_{|\bar{2}\rangle}$ on $|\bar{b}\rangle$ twice, and of course these two activations cancel each other.

We proceed with $\mbox{SUM}_{12}  |22\bar{b}\rangle = |21\bar{b}\rangle$, and $I \otimes R_{|1\bar{2}\rangle}$ activates the $R_{|\bar{2}\rangle}$ on $|\bar{b}\rangle$ which is immediately cancelled out by the $I \otimes I \otimes R_{|\bar{2}\rangle}$.

Finally $\mbox{SUM}_{12}  |21\bar{b}\rangle = |20\bar{b}\rangle$, and $I \otimes R_{|0\bar{2}\rangle}$ activates $R_{|\bar{2}\rangle}$ on $|\bar{b}\rangle$ as desired. This applies the factor of $-1$ if and only if $\bar{b}=\bar{2}$, and that's what is claimed.

\end{proof}

Using this lemma we implement the operator $R_{|20\bar{2}\rangle}$ exactly by linear recursion.

As we noted earlier, all the axial reflection operators in $n$ qutrits have the same $R$-count.

Denote this $R$-count by $\mbox{rc}(n)$.

\begin{observ} \label{obs:cost:axial:reflection}
$\mbox{rc}(n)=\Theta((2+\sqrt{5})^n)$ when $n\rightarrow \infty$.
\end{observ}

\begin{proof}
We have $\mbox{rc}(1)=1, \mbox{rc}(2)=4$ (see Observation \ref{observ:CFlip}). The recurrence $\mbox{rc}(n)=4\,\mbox{rc}(n-1)+\mbox{rc}(n-2), \mbox{rp}(1)=1, \mbox{rc}(2)=4$ can be solved in closed form as $\mbox{rc}(n)=((2+\sqrt{5})^n-(2-\sqrt{5})^n)/(2 \, \sqrt{5})$.
Because $|2-\sqrt{5}|<1$ the $-(2-\sqrt{5})^n$ term is asymptotically insignificant.
\end{proof}

Thus the cost of the above exact implementation of the $n$-qutrit axial reflection operator is exponential in $n$. This defines several tradeoffs explored in the following sections.

\section{Ancilla-free reflection-based universality}

Consider integer $n\geq 1$.

For the duration of this section we set $N=3^n$.

\begin{lemma} \label{lem:general:diagonal:unitary}
Given a diagonal unitary $D \in U(N)$ and arbitrarily small $\varepsilon>0$ there is an effectively synthesizable $\varepsilon$-approximation  of $D$ composed of a global phase factor, at most $2 \, (N-1)$ axial reflection operators, and metaplectic local gates with the total $R$-count that is

1) $16 \, (\log_3(1/\varepsilon) + O(\log(\log(1/\varepsilon))))$ when $n=1$ and,

2)  smaller than $16 \, (N-1)(\log_3(1/\varepsilon) + n + O(\log(\log(1/\varepsilon))))$ when $n>1$.
\end{lemma}

Indeed, a unitary diagonal $D$ is decomposed into a product of a global phase factor and $(N-1)$ special two-level diagonals as in corollary \ref{corol:two:level:diagonal}. Each of the latter diagonals needs to be approximated to precision $\varepsilon/(N-1)$ with
$\log_3(1/(\varepsilon/(N-1))) < \log_3(1/\varepsilon) + n$.

In \cite{WWJD} Jesus Urias offers an effective $U(2)$ parametrization of the $U(N)$ group, whereby any $U \in U(N)$ is factored into a product of at most $N(N-1)/2$ special Householder reflections and possibly one diagonal unitary.

 All reflections in that decomposition are two-level.

This immediately leads to the following

\begin{thm}{(General unitary decomposition, reflection style.)} \label{thm:new:decomposition:reflection}
Given a $U \in U(N)$ in general position and small enough $\varepsilon>0$ the $U$ can be effectively approximated up to a global phase to precision $\varepsilon$ by ancilla-free metaplectic circuit with $R$-count of at most $4\,(N+4)(N-1)(\log_3(1/\varepsilon)+ 2\, n + O(\log(\log(1/\varepsilon))))$ and at most $(N+4)(N-1)/2$ axial reflections (in $n$ qutrits).
\end{thm}

\begin{proof}
It follows from \cite{WWJD} that $U$ is effectively decomposed into $N\,(N-1)/2$ special Householder reflections and possibly a diagonal unitary $D \in U(N)$ that may add up to $2\,(N-1)$ such reflections (see Lemma \ref{lem:general:diagonal:unitary}) to the decomposition to a total of $(N+4)(N-1)/2$ reflections. Each of these allows an effective $\varepsilon/((N+4)(N-1)/2)$-approximation by a metaplectic circuit with the $R$-count of at most $8\, (\log_3(1/\varepsilon)+ 2\, n + O(\log(\log(1/\varepsilon))$ plus at most $2$ axial reflections as per Corollary \ref{corol:two:level:reflections}, and the cost bound claimed in the theorem follows.
\end{proof}

The best know cost of exact metaplectic implementation of an $n$-qutrit axial reflection is in $\Theta((2+\sqrt{5})^n)$ as per Observation \ref{obs:cost:axial:reflection}. This may become prohibitive when $n$ is large.

In the next section we show how to curb the $R$-count at the cost of roughly doubling the width of the circuits.

\section{Ancilla-assisted approximation of arbitrary unitaries}
An alternative way of implementing a two-level unitary operator is through a network of  strongly controlled gates.

For $V \in U(3)$ introduce $C^n(V) \in U(3^{n+1})$ where

$C^n(V)|j_1,\ldots,j_n,j_{n+1}\rangle = $

$
\begin{cases}
|j_1,\ldots,j_n \rangle \otimes V |j_{n+1}\rangle, & j_1 = \cdots = j_n= 2 \\
|j_1,\ldots,j_n,j_{n+1}\rangle , & \mbox{otherwise.}
\end{cases}$

The $C^1(\mbox{INC})$ gate,

\begin{equation} \label{eq:new:CINC:gate}
C^1(\mbox{INC})|j,k\rangle = |j,(k+\delta_{j,2}) \mod{3}\rangle
\end{equation}

is going to be of a particular interest in this context.

Bullock et Al. \cite{BullockEtAl} offer a certain ancilla-assisted circuit that emulates $C^n(V)$ using only two-qudit gates.

The circuit requires $n-1$ ancillary qutrits,  $4\, (n-1)$ instances of the $C^1(\mbox{INC})$ gate (see equation (\ref{eq:new:CINC:gate})) and one single $C^1(V)$ gate.

We do not believe that the classical $C^1(\mbox{INC})$ gate can be represented exactly and must resort to approximating $C^1(\mbox{INC})$ to desired precision.

\begin{lem} \label{lem:new:CINC:approximation}
$C^1(\mbox{INC})$ (as defined by  (\ref{eq:new:CINC:gate})) can be approximated to precision $\varepsilon$ by a metaplectic circuit with $R$-count at most $16 \, \log_3(1/\varepsilon) + O(\log(\log(1/\varepsilon)))$ and $2$ two-qutrit axial reflections.
\end{lem}

\begin{proof}
$C^1(\mbox{INC})$ is the composition of two reflection operators: $C^1(\mbox{INC}) = R_{|2\rangle \otimes v_{2}} \, R_{|2\rangle \otimes v_{0}}$ where
$v_{0}=(|1\rangle-|2\rangle)/\sqrt{2}, \, v_{2}=(|0\rangle-|1\rangle)/\sqrt{2}$ and the lemma follows.

\end{proof}

\begin{corol} \label{corol:new:CnV:ancilla:assisted}
Given a  $V \in U(3)$, integer $n>0$ and a small enough $\varepsilon >0$, the  $C^n(V)$ can be effectively emulated approximately to precision $\varepsilon$ by ancilla-assisted $2\, n$-qutrit circuit  with $R$-count smaller than
$64\, n \, (\log_3(1/\varepsilon)+O(\log(\log(1/\varepsilon))))$.
\end{corol}

It is easy to see from Lemma \ref{lem:two:basis:vectors} that any two-level $n$-qutrit unitary $W$ is effectively classically equivalent to some
$C^{n-1}(\tilde{W})$ where $\tilde{W}$ is a certain (two-level) single-qutrit derivative of $W$. This applies, in particular, to the two-level Householder reflections that constitute the factors in the explicit $U(2)$ factorization of $U(3^n)$  (\cite{WWJD}).

An upper bound for the cost of ancilla-assisted emulation of arbitrary $n$-qutrit unitary is summarized in the following

\begin{thm}{(General unitary decomposition, ancilla-assisted.)} \label{thm:decomposition:ancilla:assited}
Given a $U \in U(N)$ in general position and small enough $\varepsilon>0$ the $U$ can be effectively emulated up to a global phase to precision $\varepsilon$ by metaplectic circuit with $(n-2)$ ancillas and $R$-count smaller than $32\,(N+4)(N-1)(n-1)(\log_3(1/\varepsilon)+ 2\, n + O(\log(\log(1/\varepsilon))))$.
\end{thm}

\begin{proof}
We can still exactly and effectively decompose $U$ into a global phase and at most $(N+4)(N-1)/2$ two-level Householder reflections (see the proof of Thm \ref{thm:new:decomposition:reflection}).

But now we treat each two-level reflection as classical equivalent of a $C^{n-1}(V)$ where $V$ is a single-qutrit unitary. We emulate the each reflection as such using Corollary \ref{corol:new:CnV:ancilla:assisted} and the cost bound for the overall decomposition follows.
\end{proof}

This synthesis procedure is summarized as pseudocode in Algorithm \ref{alg:decomp:multi:qutrit:ancilla} below.

\begin{algorithm}[H]
\caption{Ancilla-assisted decomposition of a general unitary.}
\label{alg:decomp:multi:qutrit:ancilla}
\algsetup{indent=2em}
\begin{algorithmic}[1]
\REQUIRE{$U \in U(3^n)$, $\varepsilon>0$}
\STATE {$U = D \,\prod_{k=1}^{K} U_k$ as per \cite{WWJD}}
\COMMENT{Diagonal $D$ and two-level $U_k$}
\STATE{$\mbox{ret} \gets \mbox{decomposition}(D,\varepsilon)$ as per Corol. \ref{corol:new:CnV:ancilla:assisted}}
\FOR{$k=1..K$}
\STATE{$c \gets \mbox{decomposition}(U_k,\varepsilon)$ as per Corol. \ref{corol:new:CnV:ancilla:assisted}}
\STATE{$\mbox{ret} \gets \mbox{ret} \, c$}
\ENDFOR
\RETURN {$ret$ }

\end{algorithmic}
\end{algorithm}

\section{The overall synthesis algorithm flow.}

Assuming ancillary qutrits are readily available, a decision point on choosing between the ancilla-free and ancilla-assisted decomposition strategies is defined by relative magnitudes of
$(2+\sqrt{5})^n$ and $64 \, n \, \log_3(1/\varepsilon)$. Comparison of the upper bounds suggests that in practice the ancilla-free solution becomes prohibitively costly when $n > 7$. Otherwise the decision threshold in $\varepsilon$ is of the form $\varepsilon_n = \Omega(3^{-(2+\sqrt{5})^n/(64\,n)})$.

The two strategies can be run in parallel on a classical computer with the best resulting circuit post-selected. This approach is shown schematically in Figure \ref{fig:new:parallel:flow}.

\begin{figure}[bt]
\includegraphics[width=3.5in]{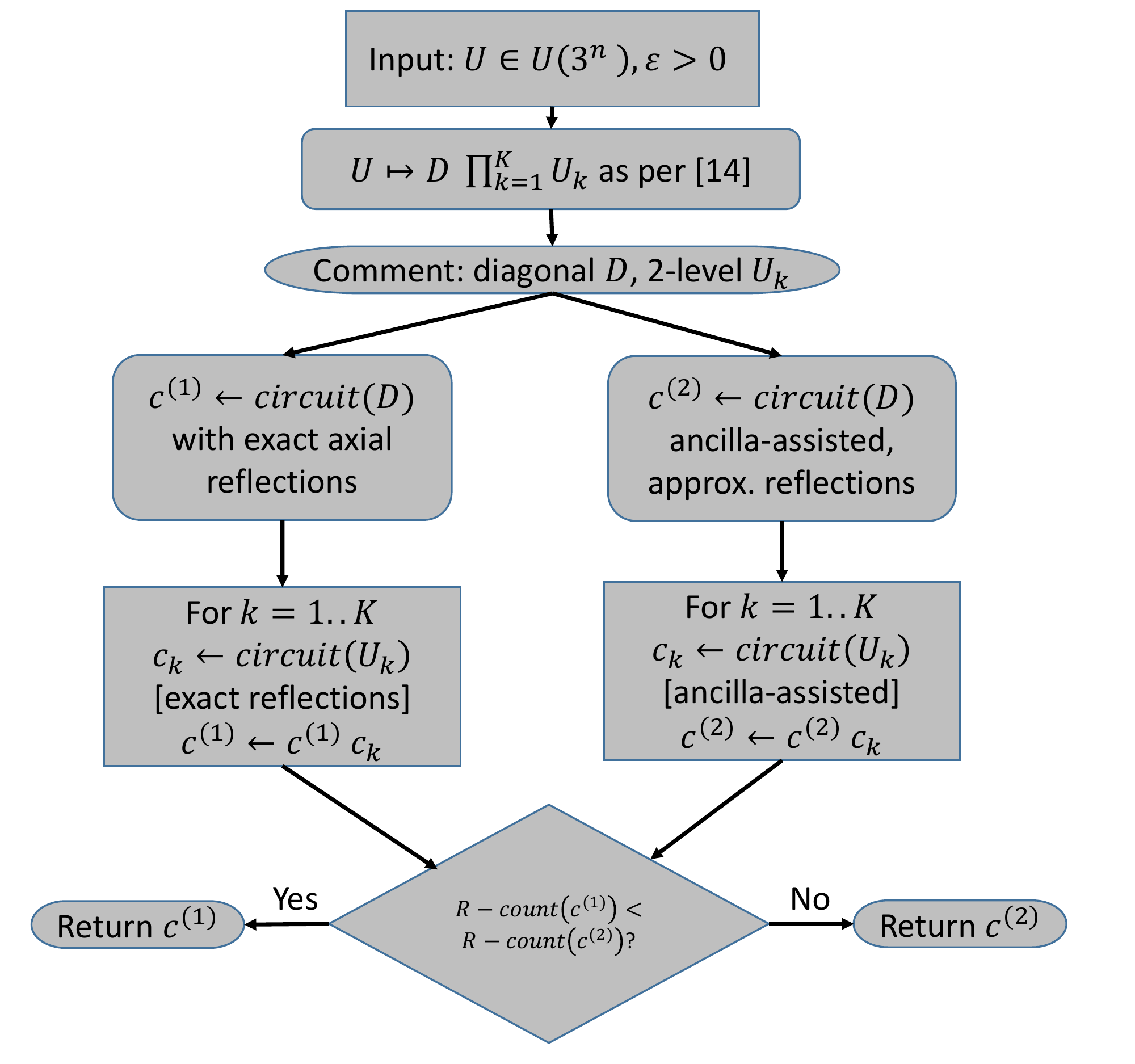}
\caption{\label{fig:new:parallel:flow} Parallelizable control flow for the two flavors of the main algorithm.}
\end{figure}

\section{Simulation, theoretical lower bound and future work.} \label{sec:bound:future}
The scaling of the cost of our metaplectic circuits is fully defined by the cost of approximating a two-level state. The $R$-count of a circuit performing an $\varepsilon$-approximation of the latter is in its turn defined by the denominator exponent $k$ of an approximating tri-level Eisenstein state $|\phi_k\rangle =(u \, |j\rangle + v \, |\ell\rangle + w \, |m\rangle)/\sqrt{-3}^k$.
.

We currently have $k$ upper-bounded by $4 \, \log_3(1/\varepsilon) + O(\log(\log(1/\varepsilon)))$.

Our numerical simulation over a large set of randomly generated two-level targets, demonstrates that an approximation algorithm based solely on Lemma \ref{lem:core:two:level:state} yields $k$ extremely close to this upper bound in overwhelming majority of cases.

A certain volume argument suggests a uniform lower bound for $k$ in
$5/2 \, \log_3(1/\varepsilon) + O(\log(\log(1/\varepsilon)))$. Indeed for a given two-level target state $|\psi\rangle$ and its $\varepsilon$-approximation $|\phi_k\rangle$ the real vector
$[Re(u),Im(u),Re(v),Im(v)]^T$ is found in a certain 4-dimension meniscus of 4-volume
$\Theta(\varepsilon^5\, 3^{2\,k})$. If we expect, uniformly, each of these menisci to contain
$\Theta(\log(1/\varepsilon))$ such vectors we need to have $\varepsilon^5\, 3^{2\,k}$ in $\Theta(\log(1/\varepsilon))$ and the above lower bound on $k$ follows.

There is clearly a gap between our guaranteed cost leading term $4 \, \log_3(1/\varepsilon)$ and the cost' lower bound leading term $5/2 \, \log_3(1/\varepsilon)$ and we currently do not know whether (a) the lower bound is reachable at all using metaplectic circuits or, (b) if it is reachable, whether this can be done by a classically tractable algorithm. More theoretical (and possibly, simulation) work is needed to answer these questions. At stake here is potential practical reduction of the metaplectic circuitry cost by $37.5\%$.

Another important open question is whether there is a set of exact metaplectic circuits for $n$-qutrit axial reflections with the $R$-count that is sub-exponential (preferrably, polynomial) in $n$.

\section{Conclusion}
We have addressed the problem of performing efficient quantum computations in a framework where quantum information is represented in multi-qutrit encoding by ensembles of certain weakly-integral anyons and the native quantum gates are represented by braids with a targeted use of projective measurement.

We have developed two flavors of a classically feasible algorithm for the synthesis of efficient metaplectic circuits that approximate arbitrary $n$-qutrit unitaries to a desired precision $\varepsilon$.
The first flavor of the algorithm produces circuits that are ancilla-free and asymptotically optimal in $\varepsilon$ (but may have additive entanglement overhead that is exponential in $n$).
The second flavor produces circuits requiring roughly $n$ clean ancillas, has a depth overhead factor of approximately $n$, but may be, nevertheless, more efficient in practice when $n$ is large.
The combined algorithm enables us to compile logical multi-qutrit circuits with the scalability properties comparable to the scalability of the recent crop of efficient logical circuits over multi-qubit bases such as Clifford+T, Clifford+V or Fibonacci.

In summary, we have demonstrated that circuit synthesis for a prospective ternary topological quantum computer based on weakly-integral anyons can be done effectively and efficiently. This implicitly validates such prospective computer for the quantum algorithm development.

Although we have achieved asymptotic optimality of the resulting circuits, there is some potential slack left in the practical bounds of leading coefficients for the circuit depths, as explained in the section \ref{sec:bound:future}. Investigating this presumed slack is one of our future research topics.


\acknowledgements
The authors wish to thank Martin Roetteler for useful discussions.


\appendix

\section{Exact Representation of Single-Qutrit Unitaries over the Metaplectic Basis} \label{sec:exact:representation}

Surprisingly, our synthesis algorithms did not require a usual theorem regarding exact decomposition of exactly representable matrices. For completeness we state such result here (Thm \ref{thm:exact:single:qutrit:synthesis}).

\begin{lemma} \label{lem:short:column:case2}
Let $|\psi\rangle$ be a unitary single-qutrit state of the form $|\psi\rangle =1/\sqrt{-3}^L ( v \, |1\rangle + w \, |2\rangle)$ where $v,w \in \Z[\omega], L \in \Z$.
Then $|\psi\rangle$ is effectively and immediately reducible to a standard basis vector at the cost of at most one $P$ gate.
\end{lemma}

\begin{proof}
We reuse remarks in the proof of lemma \ref{lem:core:short:column} to note that, whenever $L>0$ then $|v|^2 \mod 3 = |w|^2 \mod 3 = 0$. This also implies that each of the $v,w$ is divisible by $1+2\, \omega = \sqrt{-3}$ in $\Z[\omega]$. Therefore, the state reduces algebraically to a unitary state of the form $v' \, |1\rangle + w'\, |2\rangle$ where $v',w' \in \Z[\omega]$ and the lemma follows for the lemma \ref{lem:zero:exponent}.
\end{proof}

\begin{thm}[Single-qutrit exact synthesis theorem] \label{thm:exact:single:qutrit:synthesis}
Consider a $3 \times 3$ unitary matrix of the form $U=1/\sqrt{-3}^L \, M$ where $M$ is a $3 \times 3$ matrix over $\Z[\omega]$. Then $U$ is represented exactly by a metaplectic circuit of $R$-count at most $L+3$.
\end{thm}

In order to prove the theorem, we handle the following special case first:

\begin{lemma} \label{lemma:spec:case2}
Consider a $2 \times 2$ unitary matrix of the form $V=1/\sqrt{-3}^L \, M$ where $M$ is a $2 \times 2$ matrix over Eisenstein integers.
The $3 \times 3$ matrix  $U=\left(\begin{array}{cc}
              1 & 0 \\
              0 & V \\
            \end{array}\right)$ can be effectively reduced to identity by application of at most two $P$ gates and at most one classical gate.
\end{lemma}

\begin{proof} (Of the lemma.)
Let $1/\sqrt{-3}^L \, [0, u, v]^T$ be the second column of the matrix $U$.
As per lemma \ref{lem:short:column:case2} the column can be reduced to a standard basis vector using at most one $P$ gate. Applying an appropriate classical gates if necessary we can force it to be $|1\rangle$ and thus $U$ gets reduced to $diag(1,1,\varphi)$ where $\varphi \in \Z[\omega]$ is a phase factor and thus an Eisenstein unit.
Hence $\varphi = (-\omega^2)^d, d \in \Z$ and $P_2^{-d \mod 6}$ completes the reduction of the matrix to identity.

\end{proof}

\begin{proof} (Of the theorem.)

As per lemma \ref{lem:core:short:column} we can effectively find a unitary circuit $c_1$ of $R$-count at most $L+1$ and $H$-count at most $L$ that reduces the first column of $U$ to a basis vector and, in fact w.l.o.g. to $|0\rangle$.

Consider the matrix $c_1\, U$. Due to unitariness, it must be of the form $\left(\begin{array}{cc}
              1 & 0 \\
              0 & V \\
            \end{array}\right)$ with  $V = 1/\sqrt{-3}^{L_1} \, M_1$ where  $M_1$ is a certain $2 \times 2$ matrix over $\Z[\omega]$.

As per lemma \ref{lemma:spec:case2} this matrix can be effectively reduced to identity at the cost of at most two $P$ gates.

Therefore we have effectively found a circuit $c_2$ with $R$-count at most $L + 3$ and $H$-count at most $L$ such that $c_2\, U = I$ and thus $U = c_2^{-1}$.

\end{proof}

\section{Single-Qutrit State approximation} \label{sec:single:qutrit:approx}

\subsection{Norm equation in Eisenstein integers} \label{subsec:norm:equation}

The ring of the Eisenstein integers $\Z[\omega]$ is arguably the simplest \emph{cyclotomic} ring (\cite{LWashington}).

In what follows we would need certain properties of the equation

\begin{equation} \label{eisen:norm:equation}
|z|^2 = n, \, n \in \Z, \, z \in \Z[\omega]
\end{equation}

The two basic facts to deal with are: (a) the equation (\ref{eisen:norm:equation}) is solvable with respect to $z$ only for some of the right hand side values; (b) the complexity of solving the equation for $z$ is no less than the complexity of factoring the integer $n$.

The first thing to note is that $|z|^2$ is multiplicative in $z$. Therefore if $|z_1|^2=n_1$ and $|z_2|^2=n_2$ then $|z_1 \, z_2|^2 = n_1 \, n_2$. Hence disregarding the integer factorization we only need to know the effective solvability of the equation when $n$ is a power of a prime number. Moreover, since for $p \in \Z$, $|p|^2=p^2$, i.e. the equation is always solvable when $n$ is a complete square, we only need the effective solvability when $n$ is a prime number.

According to \cite{LWashington}, if $n$ is a positive prime number, the equation (\ref{eisen:norm:equation}) is solvable if and only if $n=1 \mod 3$ or $n=3$.

In case of $n=3$ the six solutions of the equation are $(-\omega)^{2\,d} \, (2 \, \omega+1), \, d=0,\ldots,5$.

In the more general case when $n$ is a prime with $n=1 \mod 3$ it is easy to obtain all the solutions of (\ref{eisen:norm:equation}) at a runtime cost that is probabilistically polynomial  in  $\log(n)$.

Here is the two step procedure to be used:

1) Compute $m \in \Z$ such that $m^2=-3 \mod n$ , using, for example, Tonelli-Shanks algorithm \cite{DShanks}.

2) Compute $z = GCD_{\Z[\omega]} ( m + 2\, \omega + 1, n)$

3) Now $\{(-\omega^2)^d \, z, \, (-\omega^2)^d \, z^*, \, d=0,\ldots 5\}$ are the solutions of (\ref{eisen:norm:equation}).

As a matter of principle we could limit ourself only to norm equations with integer prime right hand sides and thus sidestep the need for integer factorization.

If we pick an integer $n$ at random from some interval $(B/2,B)$, then the probability that $n$ is an integer prime with $n=1 \mod 3$ is going to be in $\Omega(1/\log(B))$ (cf. \cite{HazePrime}).

While it is sufficient for establishing asymptotic properties of the algorithms we are about to design, for improved practical performance it is beneficial to be able to deal with \emph{easily solvable} equations of the form (\ref{eisen:norm:equation}), that is the ones where the integer $n$ on the right hand side can be factored at some acceptable cost. A subset of solutions of the equation in this case is described by the following

\begin{thm} \label{thm:norm:eq:composite}
Let $n$ be an integer, factored to the form $n = m^2 \, p_1\, \cdots p_{\ell}$ , where $m \in \Z$ and $p_1\, \ldots p_{\ell}$ are distinct positive integer primes.

Then

1) The equation (\ref{eisen:norm:equation}) is solvable if and only if $p_j = 1 \mod 3, j = 1,\ldots,\ell$.

2) If $\{z_1,\ldots,z_{\ell}\}$ is a sequence of particular solutions of the equations $|z_j|^2 = p_j, \, j = 1,\ldots,\ell$ then all of the following are solutions of the equation (\ref{eisen:norm:equation}):

\begin{equation}
z = m \, Conj^{d_1}[z_1] \, \cdots Conj^{d_{\ell}}[z_{\ell}], d \in \{0,1\}^{\ell}
\end{equation}

where $Conj$ is the complex conjugation operator.

\end{thm}

Recall that an integer is \emph{smooth} if it does not have prime factors above certain size \cite{GranvilleSmooth}.
Let us call an integer \emph{semi-smooth} if it is a product of a smooth integer and at most one larger prime number.

In view of the theorem and the above effective procedure for solving a norm equation with a prime right hand side, solving a norm equation with semi-smooth right hand side $n$ is easy and can be effectively performed at the runtime cost that is polynomial in $\log(n)$.

The distribution of smooth integers is described by  the de Bruijn function \cite{GranvilleSmooth}.
Even though the density of semi-smooth numbers $n$ for which the equation (\ref{eisen:norm:equation}) is solvable in interval $(B/2,B)$, may still be in $\Omega(1/\log(B))$ asymptotically, in practice such integers are much more dense than the primes with $n=1 \mod 3$.

Intuitively, in a random stream of norm equations easily solvable norm equations are not uncommon, and for large enough $B>0$ we need to sample some $O(\log(B))$ integers $n \in (B/2,B)$ to find, with sufficiently high probability, one that is semi-smooth and such that the equation (\ref{eisen:norm:equation}) is solvable.

Approximation methods developed in the next subsection depend on the following more specific

\begin{cnj} \label{conj:norm:eq:solvability}
Let $k$ be an arbitrarily large positive integer and let $u, v \in \Z[\omega]$ be randomly picked Eisenstein integers such that

$\Theta(3^{k/2}) \leq |u|^2+|v|^2 \leq 3^k$.

Then for $n = 3^k - |u|^2 - |v|^2$ the equation (\ref{eisen:norm:equation}) is easily solvable with probability that has uniform lower bound in $\Omega(1/k)$.
\end{cnj}

\subsection{Approximation of single-qutrit states}

We start with the following

\begin{lem} \label{lem:pair:approximation}
Let $|\psi\rangle$ be a unitary state of the form $x \, |0\rangle + y\, |1\rangle, \, x,y \in \mathbb{C}, |x|^2+|y|^2=1$ and let $\varepsilon$ be small enough positive value. the unitary state $|\psi\rangle$ can be approximated to precision $\varepsilon$ by a unitary state for the form $(u \, |0\rangle+ v \, |1\rangle+ w \, |2\rangle)/\sqrt{-3}^k, \, u,v,w \in \Z[\omega],\, k \in \Z$ such that
$k \leq 4\, \log_3(1/\varepsilon)+O(\log(\log(1/\varepsilon)))$. The expected classical runtime required to do the approximation effectively is polynomial in $\log(1/\varepsilon)$.
\end{lem}

\begin{figure}[bt]
\includegraphics[width=3.5in]{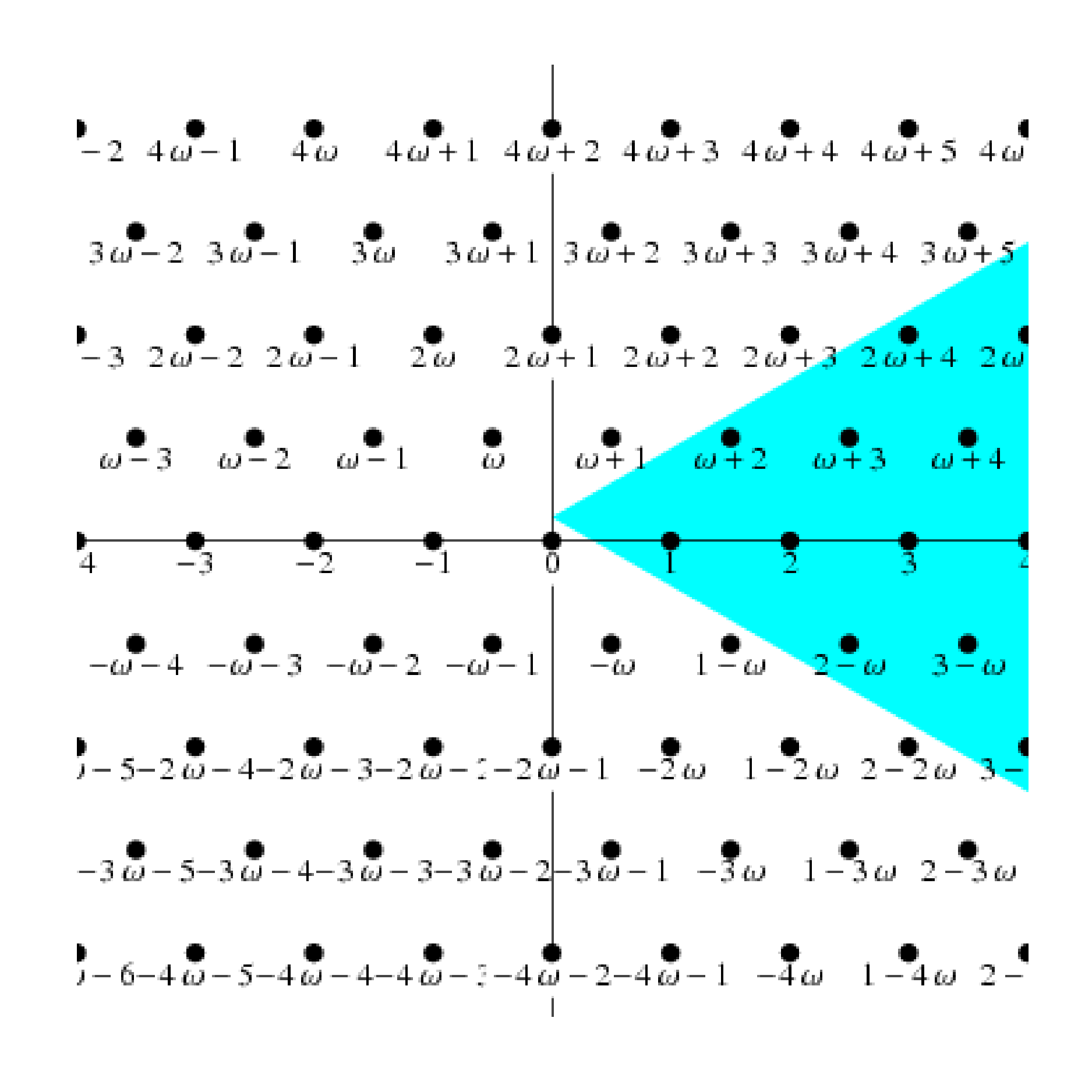}
\caption{\label{fig:eisen:integers} Lattice of Eisenstein integers. \footnote{Downloaded from http://mathworld.wolfram.com/EisensteinInteger.html, a Wolfram Research Inc. web resource.}}
\end{figure}

Before proving the lemma, let us make the following

\begin{prop} \label{approx:one:complex:revised}
For a given complex number $z$ with $|z| \leq 1$ and small enough $\varepsilon>0$ there exists an integer $k \leq 2\, \log_3(1/\varepsilon) + 5$ and an Eisensten integer $u \in \Z[\omega]$ such that $|u/\sqrt{-3}^k - z| < \varepsilon$ and $|u/\sqrt{-3}^k| \leq |z|$.

Set $k_0= \lceil 2\,\log_3(1/\varepsilon)+ 2\, \log_3(2)+2 \rceil$ and let $\ell$ be a non-negative integer that can be arbitrarily large. For $k = k_0 +\ell$ there are $\Omega(3^{\ell})$ distinct choices of Eisensten integer $u$ such that $|u/\sqrt{-3}^k - z| < \varepsilon$

\end{prop}

\begin{proof}
Note that $|u/\sqrt{-3}^k - z| = |u/\sqrt{3}^k - z\, i^k|$, and we can simplify the statement a bit by relabeling $z\, i^k$ as $z$.

We start by taking a geometric view on the feasibility of  both claims in this proposition.

On the complex plain Eisenstein integers are found at the nodes of a hexagonal lattice spanned, for example,  by $1$ and $1+\omega=+1/2 + i \, \sqrt{3}/2$. These two lattice basis vectors are at the angle $\pi/3$ (and thus the entire lattice is a tiling of the plane with equilateral triangles of side length $1$, see Figure \ref{fig:eisen:integers}). A circle of radius $R$ centered at the origin contains at least $3\,R\,(R+1)$ nodes of this lattice. As per general properties of integral lattices, a convex domain with large enough area $A$ is to contain $O(A)$ lattice nodes
and, in this case, at least $3/\pi \, A$ nodes.

Desired Eisenstein integer $u$ must be within $\varepsilon \, \sqrt{3}^k$ from $z\, \sqrt{3}^k$ and satisfy the side condition

\begin{equation} \label{eq:side:condition}
|u| \leq |z|\, \sqrt{3}^k
\end{equation}

Geometrically this means that $u$ must belong to the intersection of the two circles
$B(k,\varepsilon)= \{|u| \leq |z|\, \sqrt{3}^k\} \cap |u - z \, \sqrt{3}^k| < \varepsilon  \, \sqrt{3}^k|$.

$B(k,\varepsilon)$ is a convex domain and, when  $\varepsilon$ is sufficiently smaller than $|z|$ it contains a sector of the smaller circle with the area of at least
$1/2 \, (1-\varepsilon/|z|) \varepsilon^2 \, 3^k$. Thus (assuming $\varepsilon < 2/3 \, |z|$) if $k$ is larger than $\underline{k} = \log_3(2/\varepsilon^2) +1$ then the area of $B(k,\varepsilon)$ is greater than $1$ and $B(k)$ has a good chance of containing at least one node of the Eisenstein lattice. It may not contain one for a specific geometric configuration, but one notes that for $k=\underline{k}+\ell$ the area of $B(k,\varepsilon)$ grows exponentially in $\ell$ so there exists a small constant $\ell_0$ such that for $k_0=\lceil \underline{k} \rceil +\ell_0$ the $B(k_0, \varepsilon)$ contains an Eisenstein lattice node. It is geometrically obvious that from that point on for integer $\ell>0$ the number of Eisenstein lattice points in $B(k_0+\ell, \varepsilon)$ grows as $O(3^{\ell})$.

We now propose a procedure for effectively finding such points in $B(k, \varepsilon)$ .

The task is reduced to the case when $\pi/12 \leq \arg{z} \leq 5\,\pi/12$. Indeed, the multiplication by the Eisensten unit $-\omega^2 = 1+\omega$ is interpreted as a central rotation of the complex plane by the angle $\pi/3$ and an automorphism of the Eisenstein integer lattice. A complex number $z \neq 0$ lying in any of the six sectors $\pi/12 + \pi/3 \, m \leq \arg{z} \leq 5\,\pi/12 + \pi/3 \, m, m = 0,\ldots,5$ can be moved into the sector $\pi/12 \leq \arg{z} \leq 5\,\pi/12$ by applying zero or more of such rotations. An Eisenstein integer properly approximating the rotated target can be rotated back into an Eisenstein integer approximating the original target.

We now assume, that $k \geq \log_{\sqrt{3}}(2/\varepsilon)+2 = 2\,\log_3(1/\varepsilon)+ 2\, \log_3(2)+2$
(this is a convenient even if somewhat excessive assumption).

This implies that $\varepsilon \, \sqrt{3}^{k-1} \geq 2 \, \sqrt{3}$ and $\varepsilon \, \sqrt{3}^k \geq 6$.

 Considering $\pi/12 \leq \arg{z} \leq 5\,\pi/12$, the circle (\ref{eq:side:condition}) contains the vertical segment
$[z \, \sqrt{3}^k -  i \,|z| \sqrt{3}^k (2\, \sin(\pi/12)), z \, \sqrt{3}^k]$ of length at least $1/2 \,  |z| \sqrt{3}^k$. Assuming, again, $\varepsilon < |z|$ the
$B(k,\varepsilon/4)$ contains the vertical segment $V=[z \, \sqrt{3}^k -   i \, \sqrt{3}^k \, \varepsilon/4, z \, \sqrt{3}^k]$.

We are now ready to build the desired Eisenstein integer $u = a + b \, \omega = (a-b/2)+ i \, (b\,\sqrt{3}/2), a,b \in \Z$. We are going to chose $b$ such that $b\,\sqrt{3}/2$ is at a distance
at most $(\varepsilon/4) \, \sqrt{3}^k $ from $\mbox{Im}(z)\, \sqrt{3}^k$. As per our choice of $k$ this implies that it is necessary and sufficient for the integer $b$ to belong to a segment of length $\varepsilon \, \sqrt{3}^{k-1} /2  \geq \sqrt{3}>1$. Therefore at least one such integer exists and can be effectively picked.

Next one must find an integer $a$ such that $u=a-(b/2)+i \, (b\,\sqrt{3}/2) \in B(k,\varepsilon)$.
As per the geometric condition $\arg{z} \leq 5\,\pi/12$, the circle (\ref{eq:side:condition}) contains horizontal segment $H =[z \, \sqrt{3}^k -  |z| \sqrt{3}^k \sin(\pi/12),z \, \sqrt{3}^k]$
of length at least $1/4 \,  |z| \sqrt{3}^k$ and under $\varepsilon < |z|$ the $B(k,\varepsilon)$
contains horizontal segment $H' =[z \, \sqrt{3}^k -  \varepsilon \, \sqrt{3}^k /4,z \, \sqrt{3}^k]$.
By elementary geometric considerations $B(k,\varepsilon)$ also contains the horizontal segment
$H'' =[z \, \sqrt{3}^k -i\, b\, \sqrt{3}/2 -  3/16\, \varepsilon \, \sqrt{3}^k,z \, \sqrt{3}^k-i\, b\, \sqrt{3}/2 ]$ of length at least $3/16\, \varepsilon \, \sqrt{3}^k$.

For our choice of $k$, $3/16\, \varepsilon \, \sqrt{3}^k \geq 3/16 \times 6 > 1$. It is necessary and sufficient for the desired integer $a$ to belong to the segment
$[\mbox{Re}(z) \, \sqrt{3}^k + b/2 -  3/16\, \varepsilon \, \sqrt{3}^k,\mbox{Re}(z) \, \sqrt{3}^k + b/2]$
of length greater than $1$ as we have just seen so the desired $a$ exists and can be effectively picked.

The geometry of  this approximation procedure is shown schematically on Figure \ref{fig:approx:geometry}.

Set $k_0 = \lceil 2\,\log_3(1/\varepsilon)+ 2\, \log_3(2)+2 \rceil$. Let $\ell$ be some positive integer.
Since the geometry of the problem for $k=k_0+\ell$ is simply the geometry of the problem at $k=k_0$ scaled out by the factor of $\sqrt{3}^{\ell}$ then the segments we used above to pick the values of $b$ and $a$ are scaled out by a factor of $\Omega(\sqrt{3}^{\ell})$ and thus allow at least $\Omega(\sqrt{3}^{\ell})$ distinct choices of $b$ and at least $\Omega(\sqrt{3}^{\ell})$ distinct choices of $a$ for each choice of $b$. Therefore there are $\Omega(3^{\ell})$ distinct choices of Eisenstein integer $u$ yielding as many distinct approximations of $z$ as claimed.

\end{proof}

\begin{figure}[bt]
\includegraphics[width=3.5in]{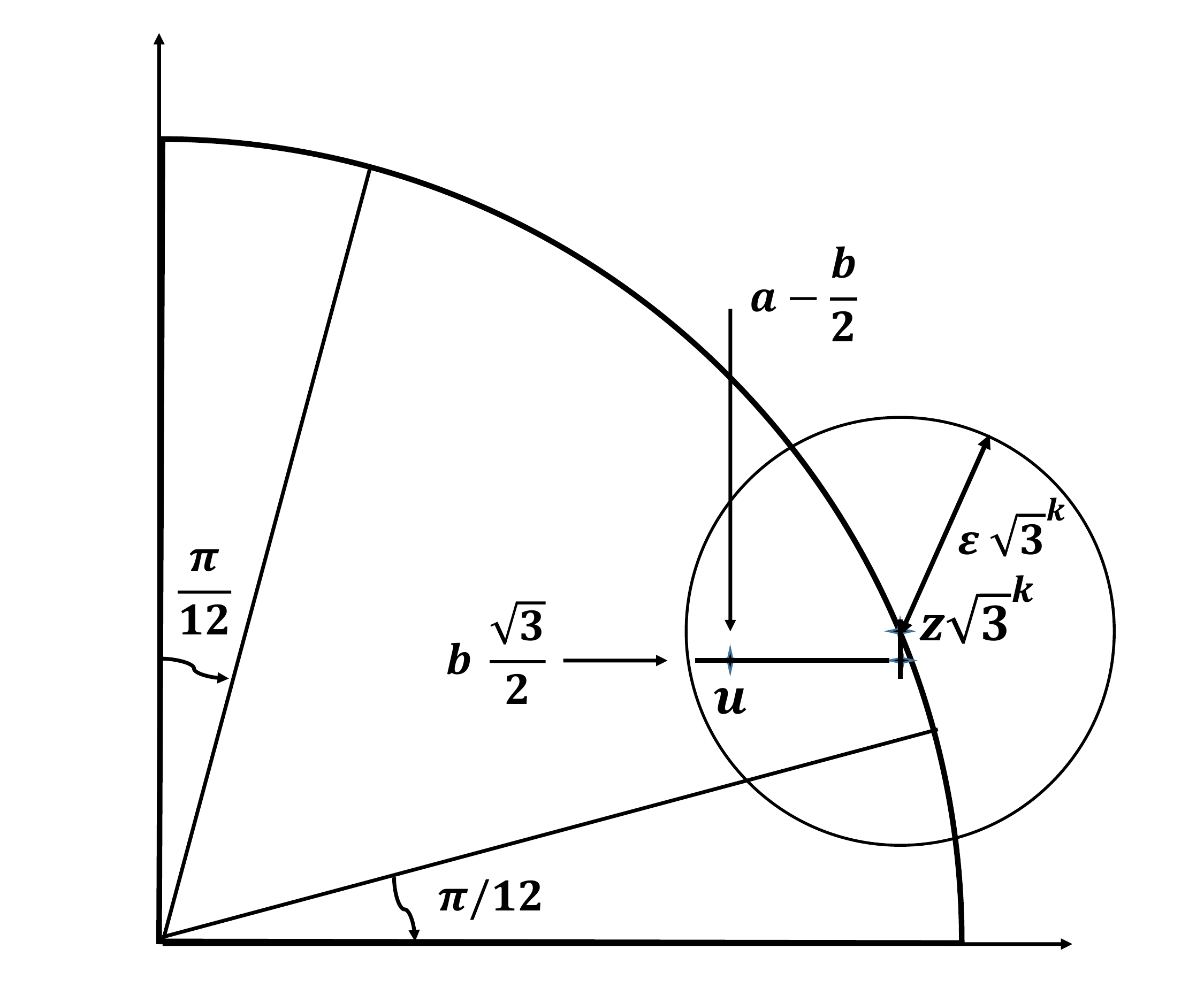}
\caption{\label{fig:approx:geometry} Approximating a scaled complex number by an Eisenstein integer.}
\end{figure}

\begin{proof}{(Of the lemma)}

For convenience we assume that $\varepsilon < 1$.

Let us do some preliminary analysis first.

We start by observing that for a unitary state  $|\phi\rangle$ to be within $\varepsilon$ of $|\psi\rangle$ , it would suffice that

\begin{equation} \label{eq:sufficient:bracket}
2 \, \mbox{Re}(\langle \phi|\psi\rangle)> 2-\varepsilon^2
\end{equation}

Consider some small $\delta > 0$ and a tri-level \emph{unitary} state $|\phi\rangle = u' \, |0\rangle + v' \, |1\rangle + w' \, |2\rangle$ and assume that
$|u'-x| <\delta, \, |v'-y| < \delta$.

By direct computation

$2\, \mbox{Re} (u' \, x^*) > |u'|^2 + |x|^2 - \delta^2$

$2\, \mbox{Re} (v' \, y^*) > |v'|^2 + |y|^2 - \delta^2$

Hence $2 \, \mbox{Re}(\langle \phi|\psi\rangle) > 2 - (1-|u'|^2-|v'|^2) - 2\, \delta^2$.

Expanding triangle inequalities $|u'|\geq |x|- |x-u'|, \, |v'|\geq |y|- |y-v'|$, we get $|u'|^2+|v'|^2 \geq 1 - 2 (|x|\,|x-u'|+|y|\,|y-v'|)+|x-u'|^2+|y-v'|^2 \geq 1- 4 \, \delta$.

Assuming w.l.o.g that $\delta^2 < \delta/2$ we conclude that $2 \, \mbox{Re}(\langle \phi|\psi\rangle) > 2 - 5\,\delta$.

Set $\delta = \varepsilon^2/5$ in order to satisfy the inequality (\ref{eq:sufficient:bracket})  and start with $k_0 = \lceil 2\,\log_3(1/\delta)+ 2\, \log_3(2)+2 \rceil \leq  4\, \log_3(1/ \varepsilon) + \log_3(5) + 5$.

We will look for a sufficient $k=k_0+\ell$ where $\ell$ iterates sequentially through non-negative integers.

As per the Proposition \ref{approx:one:complex:revised} there exist several suitable Eisenstein integers $u,v$ such that $u/\sqrt{-3}^k$ is an $\delta$-approximation of $x$ and $v/\sqrt{-3}^k$ is a $\delta$-approximation of $y$. In fact as $\ell$ grows, there are $\Omega(9^{\ell})$ distinct subunitary states $u/\sqrt{-3}^k |0\rangle + v/\sqrt{-3}^k |1\rangle$ that are $\delta$-close to $|\psi\rangle$.

To effectively prove the lemma it suffices to find one such state that can be completed to a unitary state
$|\phi\rangle = u/\sqrt{-3}^k |0\rangle + v/\sqrt{-3}^k |1\rangle + w/\sqrt{-3}^k |2\rangle$ for some $\ell$ that is not too large.

The sufficient inequality (\ref{eq:sufficient:bracket}) does not explicitly involve $w$ and is satisfied for $\delta = \varepsilon^2/5$ as shown above.

By unitariness of the desired $|\phi\rangle$, the $w \in \Z[\omega]$ must satisfy the equation

\begin{equation} \label{norm:eq:instance}
|w|^2 = 3^k - |u|^2 - |v|^2
\end{equation}

which is an instance of the norm equation (\ref{eisen:norm:equation}). As we have seen in the subsection \ref{subsec:norm:equation}, any particular instance of the norm equation is not necessarily solvable. However we are going to randomize the choice of $u$ and $v$ so that the Conjecture \ref{conj:norm:eq:solvability} becomes applicable.


To this end, let $\ell$ be an integer iterating from $0$ to some sufficiently large $L$ and $k=k_0+\ell$ iterate with it. For each subsequent value of $\ell$ we will inspect all the available $u,v$ that generate $\delta$-approximations $u/\sqrt{-3}^k, v/\sqrt{-3}^k$ of $x,y$. As we have pointed out the number of such distinct $u,v$  grows exponentially with $\ell$. Assuming  Conjecture \ref{conj:norm:eq:solvability} we only need to inspect as many as $O(\log(3^k - |u|^2 - |v|^2))=O(k)=O(k_0+\ell)$ of such distinct $u,v$  to find one for which the equation (\ref{norm:eq:instance}) is easily solvable with sufficiently high probability.

It is easy to see that there exists such $\ell = O(\log(k_0))$ for which an easily solvable norm equation (\ref{norm:eq:instance}) is obtained with near certainty. Therefore a desired unitary state $(u \, |0\rangle+ v \, |1\rangle+ w \, |2\rangle)/\sqrt{-3}^k$ will be obtained for some
$k = k_0+O(\log(k_0)) \leq
4 \, \log_3(1/\varepsilon) + O(\log(\log(1/\varepsilon)))$

Finally we note that we only needed to inspect $O(k) = O(\log(1/\varepsilon))$ candidate pairs $u,v$ for completion. Each inspection involved a decision whether the corresponding norm equation was easily solvable which incurred expected runtime cost that was polynomial in $O(\log(3^k))=O(k) = O(\log(1/\varepsilon))$, Therefore the overall expected runtime cost of the algorithm is also polynomial in $O(\log(1/\varepsilon))$.

\end{proof}

Below we present the method suggested by this lemma in pseudo-code format

\begin{algorithm}[H]
\caption{Approximation of a short state}
\label{alg:short:state:approx}
\algsetup{indent=2em}
\begin{algorithmic}[1]
\REQUIRE{$x,y \in \mathbb{C}$; $|x|^2+|y|^2=1$; $\varepsilon > 0$}
\STATE{$\delta \gets \varepsilon^2/5$}
\STATE{$k_0 \gets \lfloor 4\, \log_3(1/ \varepsilon) + \log_3(5) + 5\rfloor$}
\STATE{$w \gets \mbox{None}$; $k \gets k_0-1$}
\WHILE {$w = \mbox{None}$}
\STATE{ $k \gets k+1$ }
\STATE{ $\mbox{enum} \gets $ enumerator for all $u,v \in \Z[\omega]$}
\STATE{ s.t. $(u |0\rangle + v |1\rangle)/\sqrt{-3}^k$ is $\delta$-close to $x |0\rangle + y |1\rangle$}
\WHILE {$w = \mbox{None} \land \mbox{enum}.\mbox{Next}$}
\STATE{ $(u,v) \gets \mbox{enum}.\mbox{Current}$}
\IF{Equation $|z|^2=3^k-|u|^2-|v|^2$ is easily solvable for $z$}
\STATE{$w \gets z$}
\ENDIF
\ENDWHILE

\ENDWHILE
\RETURN { $\{u,v,w,k\}$}

\end{algorithmic}
\end{algorithm}

\section{Two-qutrit classical gates generated by $\mbox{SUM}$ and $\mbox{SWAP}$} \label{app:2qutrit:classical}

It is currently not known what two-qutrit gates can be represented exactly over the metaplectic basis. In particular it is not known whether the important classical $C^1(\mbox{INC})$
gate (\ref{eq:new:CINC:gate}) is so representable.

Let $S_9$ be the permutation group on $9$ elements.
There is a natural unitary representation of $S_9$ on $\mathbb{C}^{3^2}$ where a permutation $\pi$ is mapped to the unitary that extends the permutation $\pi$
applied to the standard basis vectors $\{|00\rangle , \ldots, |22\rangle \}$. The image of this faithful representation coincides, by definition, with the group of all the classical
two-qutrit gates. By a slight abuse of notation we also use $S_9$ to denote the image.

The following proposition addresses the maximality of the subgroup of $2$-qutrit classical gates obtained from braiding.

\begin{proposition}
The group, $G$, generated by $\mbox{SUM}, \mbox{SWAP},$ and all the $1$-qutrit classical gates is a maximal subgroup of $S_9$.
\end{proposition}

\begin{proof}
Of-course, one can always do a brute force computer search to verify this statement. Here we provide an elegant alternative proof. Let $AGL(2,\mathbb{F}_3) = GL(2,\mathbb{F}_3) \ltimes \mathbb{F}_3^2$ be the affine linear group acting on the $2$-dimensional vector space $\mathbb{F}_3^2$. Explicitly, given $\phi = (A, c) \in AGL(2,\mathbb{F}_3), v \in \mathbb{F}_3^2$, we have $\phi(v) = Av + c$. Note that $\mathbb{F}_3^2$ has in total $9$ vectors, whose coordinates under the standard basis are $\{(i,j) | i,j = 0,1,2\}$. We identify the coordinate $(i,j)$ with the $2$-qutrit basis vector $|i,j\rangle$. Since elements of $AGL(2,\mathbb{F}_3)$ permute the $9$ coordinates, we then have a group morphism $\psi: AGL(2,\mathbb{F}_3) \longrightarrow S_9 \subset U(3^2)$, such that $\psi(A,c)|i,j\rangle = A.\tvect{i}{j} + c$.

For instance, let $A = \begin{pmatrix}
1 & 0 \\
1 & 1 \\
\end{pmatrix}$, since $A.\tvect{i}{j} = \tvect{i}{i+j}$, then $\psi(A) = \mbox{SUM}$. Similarly, once can check the following correspondences.
$$
\begin{pmatrix}
1 & 0 \\
1 & 1 \\
\end{pmatrix}
\qquad
\mapsto
\qquad
\mbox{SUM}
$$
$$
\begin{pmatrix}
0 & 1 \\
1 & 0 \\
\end{pmatrix}
\qquad
\mapsto
\qquad
\mbox{SWAP}s
$$
$$
\begin{pmatrix}
1 & 0 \\
0 & 2 \\
\end{pmatrix}
\qquad
\mapsto
\qquad
Id \otimes S_{1,2}, \textrm{ where } S_{1,2} = \begin{pmatrix}
1 & 0 & 0\\
0 & 0 & 1\\
0 & 1 & 0 \\
\end{pmatrix}
$$
$$
\begin{pmatrix}
1 \\
0 \\
\end{pmatrix}
\qquad
\mapsto
\qquad
INC \otimes Id
$$
$$
\begin{pmatrix}
0 \\
1 \\
\end{pmatrix}
\qquad
\mapsto
\qquad
Id \otimes INC
$$
It's easy to check that the matrices$($vectors$)$ on the LHS of the above correspondences generate the group $AGL(2,\mathbb{F}_3)$ and the gates on the RHS generate $G$. Also, it's not hard to verify that the map $\psi$ is injective and thus $G \simeq AGL(2,\mathbb{F}_3).$ Now by $O'Nan$-$Scott$ Theorem \cite{scott}\cite{liebeck}, $AGL(2,\mathbb{F}_3)$ is a maximal subgroup of $S_9$.

Therefore $G$ is a maximal subgroup of $S_9 \subset U(3^2)$.
\end{proof}

An immediate consequence of this proposition is that, as soon as the $C^1(\mbox{INC})$ gate is   exactly representable then all the classical two-qutrit gates are also exactly representable.


\end{document}